%% file: almost-spectral-flow.tex
\documentclass[
    a4paper,
    12pt,
    DIV=11,
    headinclude=false,
    footinclude=false,
]{scrartcl}

\input{common}

\usepackage{libertinus}

\definecolorset{HTML}{col_}{}{%
    1,0072B2;%
    2,D55E00;%
    3,E69F00
}

\input{biblatex-publication}

\usepackage{xurl} 

\title{Gaussian filters in quantum lattice systems: Applications to spectral flow, local perturbations, clustering, and the quantum Hall effect}

\newcommand{\emaillink}[1]{\href{mailto:#1}{#1}}

\author{%
    Sven Bachmann%
    \texorpdfstring{%
        \,\orcidlink{0000-0002-7645-2324}
        \footnote{
            \parbox[t]{.75\textwidth}{
                Department of Mathematics,
                The University of British Columbia \\
                1984 Mathematics Road, Vancouver, BC, Canada V6T 1Z2 \\
                \emaillink{sbach@math.ubc.ca}
            }
        }
    }{}%
    \and
    Zhiqian (Simon) Du%
    \texorpdfstring{%
        \footnote{
            \parbox[t]{.75\textwidth}{
                Department of Mathematics, University of Califonia, Davis
                \\
                Davis, 95616, USA
                \\
                \emaillink{simdu@ucdavis.edu}
            }
        }
    }{}%
     \and
    Martin Fraas%
    \texorpdfstring{%
        \,\orcidlink{0000-0003-0332-8557}
        \footnote{
            \parbox[t]{.75\textwidth}{
                Department of Mathematics, University of Califonia, Davis
                \\
                Davis, 95616, USA
                \\
                \emaillink{fraas@math.ucdavis.eduu}
            }
        }
    }{}%
    \and
    Tom Wessel%
    \texorpdfstring{%
        \,\orcidlink{0000-0001-7593-0913}
        \footnote{
            \parbox[t]{.75\textwidth}{
                Fachbereich Mathematik,
                Universität Tübingen,
                \\
                Auf der Morgenstelle 10,
                72076 Tübingen,
                Germany.
                \\
                \emaillink{tom.wessel@uni-tuebingen.de}
            }
        }
    }{}%
}

\date{25 August 2025}

\newcommand{\expeq}{\stackrel{\exp L}{=}}

\begin{document}

\bgroup
\hypersetup{hidelinks}
\maketitle\thispagestyle{empty}
\egroup

\begin{abstract}
    We consider the locality and spectral properties of the smearing
    \begin{equation*}
        \tau_f(A) = \int_{-\infty}^\infty \odif{t} \, f(t) \, \tau_t(A)
    \end{equation*}
    when applied to the dynamics \(\tau_t\) of quantum spin systems.
    While recent applications of this map have used superpolynomially but not exponentially decaying functions \(f\) to ensure exact spectral properties, we use here Gaussian filters.
    This improves the locality at the expense of errors on the spectral side.
    We propose a number of concrete applications, from quasi-adiabatic continuation to correlation decay, and exponential stability away from impurities.
    Finally, we discuss an application to the quantum Hall effect.
\end{abstract}

\section{Introduction}
\label{sec:intro}

Physical time evolution flows \(t\mapsto \tau_t(A) = \evol{\I H t}{A}\) are local, where the precise meaning of locality depends on the particular setup.
Propagation is strictly within the light cone when the underlying equation is the wave equation, it is diffusive for the heat equation, and the Lieb-Robinson bound~\cite{LR1972} provides an effective ballistic propagation bound for the Schr\"odinger equation on sufficiently regular lattices.
Among others, these results are essential in proving well-posedness for these equations.
The locality that originates in the Lieb-Robinson bound has proved crucial to understand both thermal phases and (topologically ordered) ground state phases, see~\cite{hastings2010locality} for many examples.
In this context, the smearing map%
\footnote{
    When the subscript of \(\tau\) is a real number, it denotes the time evolution and when the subscript is a function, it denotes this integral.
}
\begin{equation}
    \label{eq:smearing}
    \tau_f(A) = \int_{-\infty}^\infty \odif{t} \, f(t) \, \tau_t(A)
\end{equation}
has proved particularly useful, since the decay of \(f\) controls the locality of \(A\mapsto\tau_f(A)\) while its Fourier transform \(\widehat f\) controls the spectral properties of the map.
In fact, the (Arveson) spectrum of \(\tau\) is defined as the smallest subset \(\sigma\subset\R\) such that \(\tau_f\) vanishes for all \(f\) such that \(\mathrm{supp}(\widehat f)\cap\sigma = \emptyset\), see~\cite[Definition A.1]{BDN2016}.

Although it was not phrased exactly as such in its original formulation~\cite{hastings2004locality}, smearing of the type~\eqref{eq:smearing} is central to quasi-adiabatic continuation, also known as the spectral flow~\cite{hastings2010locality,bachmann2012automorphic}, a tool that has become one of the cornerstones of the analysis and classification of topological phases of matter, and of the related questions of adiabaticity and linear response theory~\cite{bachmann2018adiabatic,monaco2019adiabatic,teufel2020non}.
Understanding quasi-particle excitations of topologically ordered states also relies on this technique~\cite{BDN2016,bachmann2020many}.
In another guise, such smearings are instrumental in proving rapid decay of correlations for gapped quantum lattice systems~\cite{hastings2006spectral,nachtergaele2006lieb}, and to relate this to stability of ground state~\cite{bachmann2022stability}.

In the example of the spectral flow, the Fourier transform \(\widehat f\) is required to be discontinuous, which implies that \(f\) cannot decay exponentially at infinity.
While stretched-exponential or even superpolynomial decay is often sufficient and convenient for locality arguments~\cite{nachtergaele2019quasi}, one may wish to use an exponentially decaying function.
In fact, the original formulation of quasi-adiabatic continuation used a Gaussian filter.
The price to pay is that the spectral properties are not exact anymore.
Exponentially decaying filter functions have also proved useful in treating thermal states, see e.g.~\cite{Hastings2007,EO2019,CMTW2023}.

In this paper, we consider the use of Gaussian filter functions, providing exact bounds both on the spatial locality and on the spectral errors.
In particular, we define a spectral flow which is exponentially local and almost exact.
Our Ansatz is similar to the one used in~\cite{hastings2005quasiadiabatic}.
Here, the smearing is applied term-by-term on an interaction and the width of the Gaussian must be chosen in a spatially inhomogeneous way.
As expected, better locality yield worse spectral mapping properties.
The same methods allow us to provide an exponential version of the \enquote{local perturbations perturb locally} property.
What is more, a properly chosen Gaussian filter convolved with a step function yields exponential decay of correlations for gapped spectral patches whose width does not need to vanish in the thermodynamic limit.
As a corollary of all the above, we conclude that the Hall conductance is quantized in finite systems up to errors that are exponentially small in the system size.

This paper is organized as follows.
After a brief review of the Lieb-Robinson bound, we introduce in Section~\ref{sec:almost-inverse-Liouvillian} the almost inverse Liouvillian built using a Gaussian filter.
We prove bounds characterizing both its quasi-locality properties and its spectral property.
In Section~\ref{sec:almost-spectral-flow}, we use it to construct an exponentially local almost spectral flow and show that local perturbations that do not close the gap have exponentially small effect away from the perturbation.
Parallel ideas are developed in Section~\ref{sec:clustering} to prove exponential clustering for finite volume spectral patches of finite width.
Finally, Section~\ref{sec:QHE} discusses the application of these tools to the quantum Hall effect.


\section{Mathematical setup}
\label{sec:mathematical-setup}

We consider spin systems on finite \(D\)-regular graphs.
Therefore, let \(\Lambda\) be a finite set and denote by \(\SYMd\) the graph distance.
Clearly, there exist constants \(D\in \N_+\) and \(\Cvol > 0\) such that the volume of all balls
\begin{gather*}
    B_x(r) := \Set{ z\in \Lambda \given \dd{z,x} \leq r }
\shortintertext{scales such that}
    \abs{B_x(r)}
    \leq
    \Cvol \, (r+1)^D
    \qquad\text{for all \(x\in \Lambda\) and \(r \geq 0\).}
\end{gather*}
The set of graphs with this scaling for fixed \(D\) and \(\Cvol\) is denoted \(\graphs\).
With this definition, all finite subsets \(\Lambda \subset \Z^D\) are in the same \(\graphs\), and they can even have periodic boundary conditions in one or more directions.

For later purposes, we note that there exists a constant \(\Ccod{b,k} \geq 1\) such that for all sets \(Z\subset \Lambda\)
\begin{equation*}
    \abs{Z}^k \, \e^{-b \diam{Z}}
    \leq
    \Ccod{b,k}
    .
\end{equation*}
If \(b \leq D\), then
\begin{equation}
    \label{eq:Cvol}
    \Ccod{b,k}
    \leq
    \Cvol^k \, \paren[\big]{\tfrac{k\,D}{\e}}^{kD} \, b^{-k D} \, \e^b
    .
\end{equation}
Indeed, for any \(z\in Z\), we have
\begin{equation*}
    \abs{Z}^k \, \e^{-b \diam{Z}}
    \leq
    \abs{B_z(\diam{Z})}^k \, \e^{-b \diam{Z}}
    \leq
    \Cvol \, \sup_{n\geq 0} \, (n+1)^{kD} \, \e^{-b n}
    ,
\end{equation*}
and the bound~\eqref{eq:Cvol} follows since the supremum is attained at \(n = \min\List{0, k \, D / b - 1}\).

With every site \(x \in \Lambda\) one associates a finite-dimensional local Hilbert space \(\Hi_x := \C^q\) with the corresponding space of linear operators denoted by \(\alg_x:=\mathcal{B} \paren*{\C^q}\).
Moreover, we define the Hilbert space \(\Hi_{\Lambda}:= \bigotimes_{x \in \Lambda} \Hi_x\), and denote the algebra of bounded linear operators on~\(\Hi_{\Lambda}\) by \(\alg_{\Lambda} := \mathcal{B}(\Hi_{\Lambda})\).
Due to the tensor product structure, we have \(\alg_{\Lambda} = \bigotimes_{x \in \Lambda} \alg_x\).
Hence, for \(X \subset \Lambda\), any \(A \in \alg_{X}\) can be viewed as an element of~\(\alg_{\Lambda}\) by identifying~\(A\) with \(A \otimes \unit_{\Lambda \backslash X} \in \alg_{\Lambda}\), where~\(\unit_{\Lambda \backslash X}\) denotes the identity in~\(\alg_{\Lambda \backslash X}\).
This identification is always understood implicitly and for \(B\in \alg_\Lambda\) we denote by~\(\supp(B)\) the smallest \(Y\subset \Lambda\) such that~\(B\in \alg_Y\).

\begin{remark}
    While we formulate all our results for spin systems, one can use the same approach for fermionic lattice systems provided the algebra is that of even elements of the CAR, see for example~\cite{NSY2017}.
    In particular, \cref{thm:almost-spectral-flow,thm:LPPL,thm:exponential clustering} also hold for fermionic lattice systems.
\end{remark}

Let \(I\subset\R\) be a compact interval.
A (time-dependent) interaction on \(\Lambda\) is a function
\begin{equation}
    \Phi \colon I\times\List{Z\subset\Lambda} \to \alg_{\Lambda},
    \quad (t,Z)\mapsto \Phi (t,Z)\in\alg_Z
    \quadtext{with} \Phi(t,Z)=\Phi(t,Z)^*
    .
\end{equation}
For \(b \geq 0\), we define interaction norms
\begin{equation}
    \label{eq:interaction-norm}
    \norm{\Phi}_b
    :=
    \sup_{t\in I} \sup_{z\in \Lambda}
    \, \sumstack{Z\subset\Lambda\suchthat\\z\in Z}
    \, \norm{\Phi(t,Z)}
    \, \e^{b \diam{Z}}
    .
\end{equation}
An interaction~\(\Phi\) gives rise to the corresponding operator
\begin{equation*}
    H(t)
    :=
    \sum_{Z \subset \Lambda} \Phi(t,Z)
    ,
\end{equation*}
which generates the Heisenberg time-evolution \(\tau_{t,s}\) defined as the solution of
\begin{equation}
    \label{eq:definition-time-evolution}
    \odv{}{t} \, \tau_{s,t}(A)
    =
    \tau_{s,t} \paren[\big]{\I\, \commutator{H(t),A}}, \qquad \tau_{s,s}(A)
    =
    A
    ,
\end{equation}
for any \(A\in \alg_\Lambda\).

An important property of the time-evolution is its locality as captured by Lieb-Robinson bounds, which originated in~\cite{LR1972} and were generalized in~\cite{NS2006} and many other works.
We here state a version for the norm~\eqref{eq:interaction-norm}, whose time-independent version appeared in~\cite[Theorem~A.1]{bachmann2022trotter}, and the time-dependent one is in~\cite[Theorem~7.3.3]{Maier2022}.

\begin{lemma}[Lieb-Robinson bound]
    \label{lem:LRB}
    Let \(D\in \N\), \(\Cvol>0\), \(\Lambda\in \mathcal{G}(D,\Cvol)\) be finite, and \(b'>b>0\).
    Then, for all intervals \(I\subset \R\), time-dependent interactions \(\Phi\) such that \(\norm{\Phi}_{b'}<\infty\), disjoint subsets \(X\), \(Y \subset \Lambda\), observables \(A\in \alg_X\) and \(B\in \alg_Y\), and \(s\), \(t\in I\) it holds that
    \begin{equation}
        \label{eq:LR-bound-general}
        \norm[\big]{\commutator[\big]{\tau_{s,t}(A),B}}
        \leq
        2 \, \Ccod{1,b'-b}^{-1} \, \norm{A} \, \norm{B}
        \, \paren[\big]{
            \e^{b \, v \, \abs{t-s}}-1
        }
        \, D(X,Y)
        ,
    \end{equation}
    where \(v = 2 \, \Ccod{1,b'-b} \, \norm{\Phi}_{b'} / b\) is the Lieb-Robinson velocity and
    \begin{align*}
        D(X,Y)
        &:=
        \min \List[\bigg]{
            \sum_{x\in X} \e^{-b\dist{x,Y}},
            \sum_{y\in Y} \e^{-b\dist{y,X}}
        }
        \\&\phantom{:}\leq
        \min \List[\big]{\abs{X},\abs{Y}}
        \, \e^{-b\dist{X,Y}}
        .
    \end{align*}
\end{lemma}

In this statement and the rest of the paper, the interaction norm and the evolution implicitly depend on~\(\Lambda\) and this dependence is understood from the context.
Importantly, all the constants will be independent of the specific \(\Lambda\in \graphs\) and only depend on the lattice through \(D\)~and~\(\Cvol\).
In fact, one can define an interaction \(\Phi\) and the interaction norm \(\norm{\cdot}_b\) on an infinite lattice \(\Gamma=\Z^D\) and the restrictions \(\Phi\rvert_{\Lambda}\) to finite lattices \(\Lambda\subset \Gamma\) satisfy \(\norm{\Phi\rvert_\Lambda}_b \leq \norm{\Phi}_b\).
Our results are then uniform in \(\Lambda\subset \Gamma\).


\section{The almost inverse Liouvillian}
\label{sec:almost-inverse-Liouvillian}

For \(\beta > 0\) let
\begin{equation*}
    \phib(t) := \frac{\beta}{\sqrt{\pi}} \, \Exp{-\beta^2 t^2}
    .
\end{equation*}
It is such that \(\int \phib = 1\) and the Fourier transform is
\begin{equation}
    \label{eq:phi hat}
    \phibhat(\omega)
    :=
    \frac{1}{\sqrt{2\pi}} \int \odif{t} \, \phib(t) \, \Exp{-\I t \omega}
    =
    \frac{1}{\sqrt{2\pi}} \, \Exp{-\frac{\omega^2}{4\beta^2}}
\end{equation}
for which \(\int \phibhat = \beta\sqrt{\frac{2}{\pi}}\).
Then, for any Hamiltonian \(H\) and observable \(A\) we define the \emph{almost inverse Liouvillian}
\begin{equation}
    \label{eq:def-almost-inverse-Liouvillian}
    \calI_{H,\beta}(A)
    :=
    \int_{-\infty}^\infty \odif{t} \, \phib(t) \, \int_{0}^{t} \odif{s} \, \tau_s^H(A)
    ,
\end{equation}
where \(\tau_s^H(A) = \evol{\I H s}{A}\).
To make contact with the previous section, we imagine here a time-dependent interaction and the corresponding Hamiltonian in a finite volume \(\Lambda\).

To quantify the properties of the almost inverse Liouvillian, it will be helpful to define the exact inverse Liouvillian \(\calI_{H}\) along the same lines, see~\cite{BMNS2012}.
Instead of \(\phib\), the map \(\calI_{H}\) uses a non-negative function \(w\) with Fourier transform \(\hat{w}\in C^1_0((-\delta,\delta))\) and \(\int w = 1\).
Note that \(w\) can be chosen to decay faster than any inverse power but not exponentially.
Whenever \(\calI_{H}\) involves a gapped Hamiltonian with gap \(\gamma\), we implicitly choose \(\delta<\gamma\).

Using standard techniques, we see that the almost inverse Liouvillian has good locality properties.

\begin{lemma}[Locality of the almost inverse Liouvillian]
    \label{lem:locality-inv-Liouvillian-beta}
    Let \(D\in \N\), \(\Cvol > 0\) and \(\Lambda\in \graphs\) be finite.
    Let \(b'>b>0\), \(\Phi\) be an interaction such that \(\norm{\Phi}_{b'}<\infty\) and \(H\) the corresponding Hamiltonian.
    Then, for all disjoint \(X\), \(Y\subset \Lambda\), \(A\in \alg_X\) and \(B\in \alg_Y\),
    \begin{equation}
        \label{eq:lem-locality-inv-Liouvillian-beta-inf}
        \norm[\big]{
            \commutator[\big]{\calI_{H,\beta}(A), B}
        }
        \leq
        2 \, \min\{\abs{X},\abs{Y}\} \, \norm{A} \, \norm{B}
        \, \inf_{T > 0}
        \, \paren[\bigg]{
            \frac{2 \, \beta}{\sqrt{\pi} \, b^2 \, v^2}
            \, \e^{b(vT-\dist{X,Y})}
            + \frac{1}{\sqrt{\pi} \, \beta}
            \, \e^{-\beta^2 T^2}
        }
        ,
    \end{equation}
    where \(v\) is the Lieb-Robinson velocity from \cref{lem:LRB}.
\end{lemma}

One may choose \(T = \frac{d(X,Y)}{2v}\) and then use \(\dist{X,Y}^2 \geq \dist{X,Y}\) to get
\begin{equation}
    \label{eq:alternative locality}
    \norm[\big]{
        \commutator[\big]{\calI_{H,\beta}(A), B}
    }
    \leq 2 \, \min\{\abs{X},\abs{Y}\} \, \norm{A} \, \norm{B}
    \paren[\bigg]{
        \frac{2 \, \beta}{\sqrt{\pi} \, b^2 \, v^2}
        + \frac{1}{\sqrt{\pi} \, \beta}}\e^{- b(\beta)d(X,Y)}
    ,
\end{equation}
where \(b(\beta) = \min\List[\big]{\frac{b}{2}, \frac{\beta^2}{4v^2}}\).

\begin{remark}
    The last expression~\eqref{eq:alternative locality}, in particular \(b(\beta)\), emphasizes the double origin of the quasi-locality of the map \(\calI_{H,\beta}\), namely the locality of the Hamiltonian through the Lieb-Robinson bound yielding the dependence on \(b\) and the locality expressed by the Gaussian filter yielding the dependence on \(\beta\).
    In particular, the locality cannot be improved further than \(\beta\) of the order of \(\sqrt{b}\), and we shall later restrict our attention to this situation, see~\eqref{eq:Hyp on beta} below.
\end{remark}

\begin{proof}
    We bound
    \begin{equation*}
        \norm[\big]{
            \commutator[\big]{\calI_{H,\beta}(A), B}
        }
        \leq
        \int_{-\infty}^\infty \odif{t} \, \phib(t) \, \int_{0}^{\mathrlap{t}} \odif{s}
        \, \norm[\big]{
            \commutator[\big]{\tau_s^H(A), B}
        }
        =
        \inf_{T>0}(I_{<T} + I_{\geq T})
        ,
    \end{equation*}
    and split the integral into a part \(I_{<T}\) where \(\abs{t}<T\) and the rest \(I_{\geq T}\) where \(\abs{t}\geq T\).
    For the first part we use \(\phi_\beta(t) \leq \beta/\sqrt{\pi}\) and the Lieb-Robinson bound for \(\tau^H\), \cref{lem:LRB}, where we bound \(\Ccod{1,b'-b}^{-1} \leq 1\), to obtain
    \begin{align*}
        I_{<T}
        &\leq
        \frac{\beta}{\sqrt{\pi}}
        \, 4 \, \min\{\abs{X},\abs{Y}\} \, \norm{A} \, \norm{B} \, \e^{-b\dist{X,Y}}
        \int_{0}^{\mathrlap{T}} \odif{t} \, \int_{0}^{\mathrlap{t}} \odif{s} \, \paren[\big]{\e^{bvs}-1}
        \\&=
        \frac{\beta}{\sqrt{\pi}}
        \, 4 \, \min\{\abs{X},\abs{Y}\} \, \norm{A} \, \norm{B}
        \, \frac{1}{(bv)^2}\paren[\big]{\e^{bvT}-1-bvT-\tfrac{1}{2}(bvT)^2}
        \, \e^{-b\dist{X,Y}}
        .
    \end{align*}
    For the second term we use the decay of \(\phi_\beta\) and the trivial bound
    \(
        \norm[\big]{
            \commutator[\big]{\tau_s^H(A), B}
        }
        \leq
        2 \, \norm{A} \, \norm{B}
    \)
    to get
    \begin{equation*}
        I_{\geq T}
        \leq
        4 \, \norm{A} \, \norm{B}
        \, \frac{\beta}{\sqrt{\pi}}
        \int_{T}^{\mathrlap{\infty}} \odif{t} \, \Exp{-\beta^2 t^2} \, t
        \leq
        2 \, \norm{A} \, \norm{B}
        \, \frac{1}{\sqrt{\pi} \, \beta}
        \, \paren[\big]{\Exp{-\beta^2 T^2} -1}
        .
    \end{equation*}
    Combining both bounds concludes the proof.
\end{proof}

\begin{assumption}[Gap Assumption]
    \label{assumption:gap}
    Let \(\Lambda\) be finite.
    Let \(\Phi\) be an interaction satisfying \(\norm{\Phi}_{b'}<\infty\) for some \(b'>0\), and let \(H\) be the corresponding Hamiltonian.
    We assume that the spectrum of \(H\) is of the form
    \begin{equation*}
        \sigma\paren{H} = \sigma_0 \cup \sigma_1
    \end{equation*}
    with \(\gamma = \dist{\sigma_1,\sigma_0}>0\).
\end{assumption}

From here onwards, we shall drop the subscript \(\Lambda\) for notational clarity.
We denote \(P = \chi_{\sigma_0}(H)\) the spectral projection of \(H\) corresponding to the patch \(\sigma_0\).
\begin{proposition}
    \label{prop:almost-inverse-Liouvillian-vanishes-on-PAP}
    Let \(H\) satisfy the \nameref{assumption:gap} and assume that \(\sigma_0 = \List{E_0}\) is a single eigenvalue.
    Then
    \begin{equation*}
        \calI_{H,\beta}( P \, A \, P )
        =
        0
        ,
    \end{equation*}
    for all \(A\in \alg_\Lambda\).
\end{proposition}

\begin{proof}
    One calculates
    \begin{equation*}
        \calI_{H,\beta}( P \, A \, P)
        =
        \int_{-\infty}^\infty \odif{t} \, \phib(t) \, \int_{0}^{\mathrlap{t}} \odif{s} \, \tau_s^H(P \, A \, P)
        =
        \int_{-\infty}^\infty \odif{t} \, \phib(t) \, t \, P \, A \, P
        =
        0
        ,
    \end{equation*}
    because the time-evolution is trivial and \(\phib\) is an even function.
\end{proof}

\begin{remark}
    The fact that \(\calI_{H,\beta}\) vanishes exactly on the range of \(P\) depends on the assumption \(\sigma_0 = \List{E_0}\).
    For spectral patches with \(\delta = \diam{\sigma_0}>0\) sufficiently small with respect to \(\gamma\), a similar result could be obtained, although with an error bound, by replacing \(\phibhat\) with two Gaussians centred at \(\pm \tfrac{\delta + \gamma}{2}\).
\end{remark}

\begin{proposition}
    \label{prop:almost-inverse-Liouvillian-bound}
    Let \(H\) satisfy the \nameref{assumption:gap}.
    Then \(\calI_{H,\beta}\) is an almost inverse of the Liouvillian \(\calL_H = -\I \, \commutator{H, \cdot}\) on off-diagonal operators.
    More precisely, for all \(A\in \alg\) such that \(A = P \, A \, P^\perp\) or \(A = P^\perp \, A \, P\), and for all \(q\in\intervalcc{1,\infty}\),
    \begin{align*}
        \norm[\big]{
            \calI_{H,\beta} \circ \calL_{H} (A)
            - A
        }_q
        &\leq
        \, \norm{P_\mu}_1
        \, \norm{A}
        \, \Exp{-\frac{\gamma^2}{4\beta^2}}
        .
    \end{align*}
    Moreover, for all \(A\in \alg\) and \(q\in \intervalcc{1,\infty}\), it holds that
    \begin{equation*}
        \norm[\big]{
            \commutator[\big]{
                \calI_{H,\beta} \circ \calL_{H} (A)
                - A
                ,
                P
            }
        }_q
        \leq
        2
        \, \norm{P_\mu}_1
        \, \norm{A}
        \, \Exp{-\frac{\gamma^2}{4\beta^2}}
        .
    \end{equation*}
\end{proposition}

\begin{proof}
    For the first statement, since the calculation is exactly analogous, we only consider the case \(A = P \, A \, P^\perp\).
    From the definition of the almost spectral flow, the fact that \(\tau^H_s (\calL_{H} (A)) = - \odv{}{s} \, \tau^H_s(A)\) and the spectral theorem for \(H\), we obtain
    \begin{align*}
        \paren[\big]{\calI_{H,\beta}- \calI_{H}}
        \circ \calL_{H} (A)
        &=
        \int_{-\infty}^\infty \odif{t} \, \paren[\big]{w(t) - \phib(t)}
        \, \paren[\big]{\tau^H_t(A) - A}
        \\&=
        \sumstack[l]{\mu \in \sigma_0}
        \, \sumstack[r]{\nu \in \sigma_1}
        \int_{-\infty}^\infty \odif{t} \, \paren[\big]{w(t) - \phib(t)}
        \, \Exp{\I t (\mu-\nu) }
        \, P_\mu \, A \, P_\nu
        \\&=
        \sqrt{2\pi}
        \, \sumstack[l]{\mu \in \sigma_0}
        \, P_\mu \, A
        \, \sumstack[r]{\nu \in \sigma_1}
        \phibhat(\nu-\mu)
        \, P_\nu
        ,
    \end{align*}
    where we used \(\int w = 1 = \int \phib\) and \(\hat{w}(\nu-\mu) = 0\) because \(\abs{\nu-\mu} \geq \gamma\) by the \nameref{assumption:gap}.
    By the triangle and Hölder inequalities, in particular using that \(\norm{P_\mu}_q \leq\norm{P_\mu}_1\) for all \(q\) and \(\mu\) and \(\sum_{\mu\in \sigma_0} \norm{P_\mu}_1 = \norm{P}_1\) we obtain,
    \begin{equation*}
        \begin{aligned}
            \norm[\big]{
                \calI_{H,\beta} \circ \calL_{H} (A)
                - A
            }_q
            &\leq
            \sqrt{2\pi}
            \, \sumstack[l]{\mu \in \sigma_0}
            \, \norm{P_\mu}_q
            \, \norm{A}
            \, \norm[\Big]{
                \sumstack[r]{\nu \in \sigma_1}
                \phibhat(\nu-\mu)
                \, P_\nu
            }
            \\&\leq
            \sqrt{2\pi}
            \, \norm{P}_q
            \, \norm{A}
            \, \phibhat(\gamma)
            .
        \end{aligned}
    \end{equation*}
    With~\eqref{eq:phi hat}, this yields the claim.

    For the second statement, note that
    \begin{equation*}
        \commutator[\big]{
            \calI_{H,\beta} \circ \calL_{H} (A)
            - A
            ,
            P
        }
        =
        \calI_{H,\beta} \circ \calL_{H}
        (\commutator{A,P})
        - \commutator{A,P}
    \end{equation*}
    because \(H\) and \(P\) commute.
    The statement follows by applying the first part since \(\commutator{A,P} = P^\perp \, A \, P - P \, A \, P^\perp\).
\end{proof}

\begin{remark}
    Since we are dealing with a fixed finite volume here, the rank of \(P\) is finite.
    For the applications we have in mind however, it is crucial to have that \(P_{\Lambda}\) remains uniformly bounded as \(\Lambda\to\Gamma\).
\end{remark}

It will later also be helpful to compare the almost inverse Liouvillian to the exact inverse Liouvillian directly.

\begin{lemma}
    \label{lem:bound-almost-inverse-Liouvillian-and-inverse-Liouvillian}
    Let \(\Lambda\) finite and \(H\in \alg_\Lambda\) satisfy the \nameref{assumption:gap}.
    For all \(A\in \alg\) such that \(A = P \, A \, P^\perp\) or \(A = P^\perp \, A \, P\), and for all \(q\in\intervalcc{1,\infty}\),
    \begin{equation*}
        \norm[\big]{
            \calI_{H,\beta}(A)
            - \calI_{H}(A)
        }_q
        \leq
        \, \norm{P}_1
        \, \norm{A}
        \, \gamma^{-1}
        \, \Exp{-\frac{\gamma^2}{4\beta^2}}
        .
    \end{equation*}
    Moreover, for all \(A\in \alg\) and \(q\in\intervalcc{1,\infty}\),
    \begin{equation*}
        \norm[\big]{
            \commutator[\big]{
                \calI_{H,\beta} (A)
                -
                \calI_H(A)
                ,
                P
            }
        }_q
        \leq
        2
        \, \norm{P}_1
        \, \norm{A}
        \, \gamma^{-1}
        \, \Exp{-\frac{\gamma^2}{4\beta^2}}
        .
    \end{equation*}
\end{lemma}

\begin{proof}
    As in the proof of \cref{prop:almost-inverse-Liouvillian-bound}, we obtain for \(A = P \, A \, P^\perp\)
    \begin{align*}
        \calI_{H,\beta}(A) - \calI_H(A)
        &=
        \sumstack[l]{\mu \in \sigma_0}
        \, \sumstack[r]{\nu \in \sigma_1}
        \int_{-\infty}^\infty \odif{t} \, \paren[\big]{\phib(t)-w(t)}
        \, \int_{0}^{t} \odif{s} \, \Exp{\I (\mu-\nu) s}
        \, P_\mu \, A \, P_\nu
        \\&=
        \sumstack[l]{\mu \in \sigma_0}
        \, P_\mu \, A
        \, \sumstack[r]{\nu \in \sigma_1}
        \, \frac{\sqrt{2\pi}}{\I (\mu-\nu)}
        \, \phibhat(\nu-\mu)
        \, P_\nu
        ,
    \end{align*}
    by the \nameref{assumption:gap}.
    It follows that
    \begin{equation*}
        \norm[\big]{
            \calI_{H,\beta}(A) - \calI_H(A)
        }_q
        \leq
        \sqrt{2\pi}
        \, \sumstack[l]{\mu \in \sigma_0}
        \, \norm{P_\mu}_q
        \, \norm{A}
        \, \norm[\Big]{
            \sumstack[r]{\nu \in \sigma_1}
            \frac{\phibhat(\nu-\mu)}{\nu-\mu}
            \, P_\nu
        }
        \leq
        \sqrt{2\pi}
        \, \norm{P_\mu}_1
        \, \norm{A}
        \, \frac{\phibhat(\gamma)}{\gamma}.
    \end{equation*}
    The second statement follows as in \cref{prop:almost-inverse-Liouvillian-bound}.
\end{proof}


\section{The almost spectral flow and LPPL}
\label{sec:almost-spectral-flow}

Under the \nameref{assumption:gap-time-dependent} below, \(\calI_{H(t)}\paren[\big]{\dot H(t)}\) generates an automorphism \(\alpha_{0,s}\) which provides a mapping between the instantaneous ground states of \(H(0)\) and \(H(s)\), namely
\begin{equation}
    \label{eq:automorphic equivalence}
    \omega_s(A)
    = \omega_0 \circ \alpha_{0,s}(A)
\end{equation}
for all \(A\in \alg_\Lambda\), where \(\omega_s(A) = \Tr(P_s)^{-1}\Tr(P_s \, A)\), see~\cite{BMNS2012}.
However, as explained in the introduction, this automorphism \(\alpha\) cannot be exponentially local.
Instead, we shall use the almost inverse Liouvillian through \(\calI_{H(t),\beta}\paren[\big]{\dot H(t)}\) to obtain an exponentially local almost spectral flow \(\alpha^\beta\).

Similarly to the locality of \(\calI_{H(t),\beta}(A)\) for strictly local operators \(A\in \alg_X\) discussed in \cref{lem:locality-inv-Liouvillian-beta}, standard arguments show that \(\calI_{H(t),\beta}\paren[\big]{\dot H(t)}\) is given by an exponentially local interaction if \(H\) and \(\dot H\) are, see~\cref{lem:interaction-for-generator-of-almost-spectral-flow}.
We first concentrate to the spectral mapping properties of~\(\alpha^\beta_{0,s}\).

\begin{assumption}[Uniform Gap assumption]
    \label{assumption:gap-time-dependent}
    Let \(H(s)\) with \(s \in \intervalcc{0,1}\) be a Hamiltonian given by a smooth time-dependent interaction \(s \mapsto \Phi(s)\) satisfying \(\sup_{s\in\intervalcc{0,1}}\norm{\Phi(s)}_{b'}<\infty\) for some \(b'>0\).
    We assume that \(H(s)\) satisfies the \nameref{assumption:gap}.
    Moreover, we assume that there exists compact intervals \(I(s)\) with endpoints depending smoothly on \(s\) such that \(\sigma_0(s) \subset I(s) \subset \R \setminus \sigma_1(s)\).
\end{assumption}

To characterize the almost spectral flow \(\alpha^\beta\), we wish to compare it to the exact spectral flow \(\alpha\), using the standard identity
\begin{equation}
    \label{eq:comparing-automorphisms}
    \alpha^1_{s,t}(A)
    - \alpha^2_{s,t}(A)
    =
    \sum_{Z\subset \Lambda}
    \int_s^t
    \odif{\lambda}
    \, \alpha^1_{s,\lambda}
    \paren[\Big]{
        \commutator[\Big]{
            \Psi_1(\lambda,Z) - \Psi_2(\lambda,Z)
            ,
            \alpha^2_{\lambda,t} (A)
        }
    }
    ,
\end{equation}
for any two automorphisms generated by interactions \(\Psi_{1}\) and \(\Psi_{2}\).
One then uses the Lieb-Robinson bound to obtain an estimate which grows in \(\supp(A)\) but is not extensive in~\(\Lambda\).
As outlined above, the locality of the interaction \(\calI_{H(t)}\paren[\big]{\dot H(t)}\) is insufficient for our purpose as \(\calI_{H(t)}\) is not exponentially local.
An exponential bound can be obtained directly using \cref{lem:bound-almost-inverse-Liouvillian-and-inverse-Liouvillian} with \(A\) being replaced by the full \(\dot H(t)\), but the resulting bound on \(\norm{\alpha_{0,s}(A)- \alpha^\beta_{0,s}(A)}\)
is proportional to \(\sup_{s\in\intervalcc{0,1}} \norm{\dot H(s)}\), which in general is extensive in \(\Lambda\).

The solution we propose below is to construct another flow \(\alpha^{\beta,X}_{t,s}\) where \(X = \supp(A)\), which is adapted to the support of the observable, and use it as an intermediate to compare the exact spectral flow \(\alpha\) with its exponentially local but approximate cousin \(\alpha^\beta\).
The generator of \(\alpha^{\beta,X}\) is given by
\begin{equation}
    \label{eq:generater-local-approximation-almost-spectral-flow}
    \sum_{Z\subset \Lambda}
    \calI_{H(t),\beta_{X,Z}}
    \paren[\big]{
        \dot \Phi(t,Z)
    }
    ,
\end{equation}
where the width of the Gaussian is modulated as follows
\begin{equation}
    \label{eq:definition-beta-X-Z}
    \frac{1}{\beta_{X,Z}^2}
    =
    \frac{1}{\beta^2}
    + \unit_{\dist{X,Z} \geq \ell} \, \dist{X,Z}
    ,
\end{equation}
and the parameter \(\ell\) will be chosen appropriately later.
We note that a very similar choice was also proposed in the original~\cite{HW2005}.
With this, we shall prove the following.

\begin{theorem}
    \label{thm:almost-spectral-flow}
    Let \(D\in \N\), \(\Cvol > 0\), \(b'>b>0\), \(C\supint>0\), \(\gamma>0\).
    Then there exist constants \(C\) and \(c > 0\) such that the following holds.
    For all \(\Lambda\in \graphs\) finite, smooth Hamiltonians \(H\) that satisfy the \nameref{assumption:gap-time-dependent} with gap \(\gamma\) and are given by interactions \(\Phi\) such that \(\norm{\Phi}_{b'}<C\supint\) and \(\norm{\dot \Phi}_{b'}<C\supint\), the flow \(\alpha^\beta_{s,0}\) generated by \(\calI_{H(s),\beta}\paren[\big]{\dot H(s)}\) is an almost spectral flow in the sense that
    \begin{equation}
        \abs[\big]{
            \omega_s(A)
            - \omega_0 \circ \alpha^\beta_{s,0}(A)
        }
        \leq
        C
        \, \abs{X}^2
        \, \norm{A}
        \, \e^{-c \beta^{-2}}
        ,
    \end{equation}
    for all \(X\subset \Lambda\) and \(A\in \alg_X\) and
    \begin{equation}
        \label{eq:Hyp on beta}
        \beta\in\intervaloo[\Big]{0,\min\List[\big]{1,\sqrt{2b}\,v}}
        ,
    \end{equation}
    where \(v\) is the Lieb-Robinson velocity from \cref{lem:LRB}.
\end{theorem}

Besides proving automorphic equivalence~\eqref{eq:automorphic equivalence} itself, the spectral flow can also be used to prove the local perturbations perturb locally (LPPL) principle.
In this case, \(\dot H(s)\in \alg_{\Lambdap}\) is strictly localized in a perturbation region \(\Lambdap \subset \Lambda\).
The strategy sketched above will yield the following result.

\begin{theorem}
    \label{thm:LPPL}
    Let \(D\in \N\), \(\Cvol > 0\), \(b'>0\), \(C\supint>0\), \(\gamma>0\).
    Then there exist constants \(C\) and \(c > 0\) such that the following holds.
    For all \(\Lambdap \subset \Lambda\in \graphs\) finite, smooth Hamiltonians \(H\) that satisfy the \nameref{assumption:gap-time-dependent} with gap \(\gamma\) and are given by interactions \(\Phi\) such that \(\norm{\Phi}_{b'}<C\supint\), \(\norm{\dot \Phi}_{b'}<C\supint\), and \(\dot\Phi(Z) = 0\) unless \(Z\subset \Lambdap\),
    \begin{equation}
        \label{eq:LPPL}
        \abs[\big]{
            \omega_s(A)
            - \omega_0(A)
        }
        \leq
        C
        \, \abs{X}^2
        \, \norm{A}
        \, \e^{-c \dist{X,\Lambdap}}
        ,
    \end{equation}
    for all \(s\in \intervalcc{0,1}\), \(X\subset \Lambda\) and \(A\in \alg_X\).
\end{theorem}

\begin{remark}
    A few remarks are in order.
    \begin{thmlist}
        \item
            The uniform gap assumption is of course crucial for the result and in general very difficult to verify, except for weak perturbations of frustration-free systems, see~\cite{michalakis2013stability,nachtergaele2022quasi}.
        \item
            An LPPL for ground states of gapped quantum spin systems was originally proposed in~\cite{BMNS2012}, but the use of the exact spectral flow there meant that the error in~\eqref{eq:LPPL} could not be proved to be exponential.
            An exponential decay was obtained in~\cite{DeRoeckLPPL} for Hamiltonians that are perturbations of free spins but under no additional gap assumption.
            We refer to the discussion in~\cite{DeRoeckLPPL} for previous perturbative results of similar nature.
        \item
            The constants do not depend on \(\Lambdap\).
            In particular, in the scenario of an increasing sequence of volumes \(\Lambda\to\Gamma\), the perturbation may be extensive.
    \end{thmlist}
\end{remark}

\subsection*{Proofs}

Within the proofs, we use \(C>0\) and \(c>0\) as generic constants that might change from line to line.
They can depend on the same parameters that determine \(C\) and \(c\) in the statements.
Typically, they depend on \(D\), \(\Cvol\), \(b\) and \(b'\) but not on \(\Lambda\).

As outlined before, one crucial ingredient to the construction of the almost spectral flow is locality of the automorphism \(\alpha^\beta\).
For the Lieb-Robinson bound from \cref{lem:LRB} to be sufficiently sharp, we first provide an exponentially local interaction for the generator \(\calI_{H(t),\beta}\paren[\big]{\dot H(t)}\).
With the locality provided by Lemma~\ref{lem:locality-inv-Liouvillian-beta}, the proof is rather standard.

\begin{lemma}
    \label{lem:interaction-for-generator-of-almost-spectral-flow}
    Let \(D\in \N\), \(\Cvol > 0\), \(D>b'>b>0\).
    Then there exists a constant \(C > 0\) such that the following holds.
    For all finite \(\Lambda\in \graphs\), \(\beta>0\), smooth Hamiltonians \(H\) given by interactions \(\Phi\) such that \(\norm{\Phi}_{b'}<\infty\) and \(\norm{\dot \Phi}_{b'}<\infty\), there exists an interaction \(\Psi_\beta\) such that
    \begin{equation*}
        \calI_{H(t),\beta}\paren[\big]{\dot H(t)}
        =
        \sum_{Z\subset \Lambda} \Psi_\beta(Z)
        \qquadtext{and}
        \norm{\Psi_\beta}_{b(\beta)/3} < C(\beta) \, \norm{\dot \Phi}_{b'}
        ,
    \end{equation*}
    where
    \begin{equation*}
        b(\beta)
        =
        \min\left\{\frac{b}{2}, \frac{\beta^2}{4v^2}\right\}
        ,\qquad
        C(\beta) = 1 + C \, b(\beta)^{-(D+1)} (\beta + \beta^{-1})
        ,
    \end{equation*}
    and \(v\) is the Lieb-Robinson velocity of \(\Phi\) as defined in \cref{lem:LRB}.
\end{lemma}

\begin{proof}
    Fix \(t\) and let \(\Omega\subset \Lambda\) and \(O \in \alg_\Omega\).
    Then, denote \(\calI_\beta(O) := \calI_{H(t),\beta}(O)\) and let
    \begin{align*}
        \Delta_0(O)
        &:=
        \cexp[\big]{\Omega}{\calI_\beta(O)}
        \in \alg_{\Omega}
    \shortintertext{and}
        \Delta_k(O)
        &:=
        \cexp[\big]{\Omega_k}{\calI_\beta(O)} - \cexp[\big]{\Omega_{k-1}}{\calI_\beta(O)}
        \in \alg_{\Omega_k}
        \quadtext{for}
        k
        \geq
        1
        .
    \end{align*}
    Here, \(\cexp{Z}{}\) denotes the conditional expectation as defined in~\cite{NSY2019}, which is a standard tool to approximate almost local operators by strictly local ones.
    Then \(\calI_\beta(O) = \sum_{k=0}^\infty \Delta_k(O)\) and the sum is finite since eventually \(\Omega_k = \Lambda\).
    By the properties of the conditional expectation
    \begin{equation*}
        \norm{\Delta_0(O)} \leq \norm{O}
    \end{equation*}
    and together with \cref{lem:locality-inv-Liouvillian-beta} and~\eqref{eq:alternative locality}
    \begin{equation*}
        \norm{\Delta_k(O)}
        \leq
        \norm[\big]{
            \paren[\big]{
                \id - \cexp[\big]{\Omega_{k-1}}{}
            }
            \, \calI_\beta(O)
        }
        \leq
        \tilde{C}(\beta) \, \abs{\Omega} \, \norm{O}
        \, \e^{-b(\beta) \, k}
        ,
    \end{equation*}
    where
    \begin{equation*}
        \tilde{C}(\beta)
        =
        2 \, \paren[\bigg]{
            \frac{2 \, \beta}{\sqrt{\pi} \, b^2 \, v^2}
            + \frac{1}{\sqrt{\pi} \, \beta}
        }
        .
    \end{equation*}
    Then, \(\Psi_\beta\) can be chosen as
    \begin{equation*}
        \Psi_\beta(Z)
        =
        \sum_{k=0}^\infty
        \, \sumstack[r]{Y\subset \Lambda\suchthat\\Y_k=Z}
        \, \Delta_k\paren[\big]{\dot \Phi(Y)}
        .
    \end{equation*}
    To estimate the interaction norm \(\norm{\Psi_\beta}\), for any \(z\in \Lambda\) we bound
    \begin{align*}
        \sumstack{Z\subset \Lambda\suchthat\\z\in Z}
        \, \norm{\Psi_\beta(Z)}
        \, \e^{p \diam{Z}}
        &\leq
        \sumstack[l]{Z\subset \Lambda\suchthat\\z\in Z}
        \sum_{k=0}^\infty
        \, \sumstack[r]{Y\subset \Lambda\suchthat\\Y_k=Z}
        \, \norm[\big]{\Delta_k\paren[\big]{\dot \Phi(Y)}}
        \, \e^{p \diam{Z}}
        \\&=
        \sum_{k=0}^\infty
        \, \sumstack[r]{Y\subset \Lambda}
        \unit_{z\in Y_k}
        \, \norm[\big]{\Delta_k\paren[\big]{\dot \Phi(Y)}}
        \, \e^{p \diam{Y_k}}
        .
    \end{align*}
    The \(k=0\) term is bounded by \(\norm{\dot \Phi}_p \leq \norm{\dot \Phi}_{b'}\).
    For \(k \geq 1\) and \(z\in Y_k\), there exists \(y\in B_z(k) \cap Y\) such that
    \begin{align*}
        \Alignindent
        \, \sumstack[r]{Y\subset \Lambda}
        \unit_{z\in Y_k}
        \, \norm[\big]{\Delta_k\paren[\big]{\dot \Phi(Y)}}
        \, \e^{p \diam{Y_k}}
        \\&\leq
        \sumstack[l]{y\in B_z(k)}
        \, \sumstack[r]{Y\subset \Lambda\suchthat\\y\in Y}
        \, \norm[\big]{\Delta_k\paren[\big]{\dot \Phi(Y)}}
        \, \e^{p \diam{Y}}
        \, \e^{2 p k}
        \\&\leq
        \tilde{C}(\beta)
        \, \e^{(2 p - b(\beta)) \, k}
        \sumstack[l]{y\in B_z(k)}
        \, \sumstack[r]{Y\subset \Lambda\suchthat\\y\in Y}
        \, \abs{Y}
        \, \norm{\dot \Phi(Y)}
        \, \e^{p \diam{Y}}
        \\&\leq
        \tilde{C}(\beta)
        \, \e^{(2 p + 2 \epsi - b(\beta)) \, k}
        \, \Ccod{\epsi,1}
        \, \Ccod{b'-p,1}
        \, \norm{\dot \Phi}_{b'}
        ,
    \end{align*}
    for any \(\epsi>0\) and \(p<b'\).
    The sum over \(k\) is then finite for \(p < b(\beta) / 2\), and we choose \(p=b(\beta)/3\) and \(\epsi = b(\beta)/12\).
    Hence, the constant \(C(\beta)\) from the statement is
    \begin{equation*}
        1
        + \sum_{k=1}^\infty
        \, \tilde{C}(\beta)
        \, \e^{- b(\beta) \, k / 6}
        \, \Ccod{b(\beta)/12,1}
        \, \Ccod{b'-b(\beta)/3,1}
        .
    \end{equation*}
    By definition, \(b(\beta)\in \intervaloo{0,b/2}\).
    Thus, \eqref{eq:Cvol} implies that \(\Ccod{b'-b(\beta)/3,1}\) is uniformly bounded in \(\beta\), while \(\Ccod{b(\beta)/12,1} \leq C \, b(\beta)^{-D}\).
    Moreover, \(\tilde{C}(\beta) \leq C \, (\beta + \beta^{-1})\), and we conclude with \(\sum_{k=1}^\infty \e^{-b(\beta) k / 6} \leq \frac{6}{b(\beta)}\).
\end{proof}

We can now continue with the proof of the almost spectral flow.
First of all, we note the following simple bounds: For every \(p\), \(c>0\)
\begin{equation}
    \label{eq:bound-polynomial-times-exponential}
    \sup_{r \geq 0} \, r^p \, \e^{-c r}
    \leq
    \paren[\big]{\tfrac{p}{\e}}^p \, c^{-p}
    ,
\end{equation}
and for every \(c>0\) and \(L\)
\begin{equation}
    \label{eq:bound-sum-exponential}
    \sum_{n \geq L} \e^{-c n}
    \leq
    \frac{\e^{c}}{c}
    \, \e^{-c \ceil{L}}
    .
\end{equation}

\begin{proof}[Proof of \cref{thm:almost-spectral-flow}]
    Let \(\alpha\) be the exact spectral flow generated by \(\calI_{H(t)}\paren[\big]{\dot H(t)}\) such that \(\omega_s = \omega_0 \circ \alpha_{0,s}\) and \(\alpha^\beta\) be the almost spectral flow.
    Moreover, fix \(X\subset \Lambda\) and let \(\alpha^{\beta,X}\) be the automorphism generated by~\eqref{eq:generater-local-approximation-almost-spectral-flow} where we choose \(\beta_{X,Z}\) as in~\eqref{eq:definition-beta-X-Z} with \(\ell\) to be chosen later.
    We then use triangle inequality to bound
    \begin{equation}
        \label{eq:proof-almost-spectral-flow-triangle-inequality-with-alpha-beta-X}
        \abs[\big]{
            \omega_s(A)
            - \omega_0 \circ \alpha^\beta_{0,s} (A)
        }
        \leq
        \abs[\big]{
            \omega_0 \circ \paren[\big]{\alpha_{0,s}-\alpha^{\beta,X}_{0,s}} (A)
        }
        + \abs[\big]{
            \omega_0 \circ \paren[\big]{\alpha^{\beta,X}_{0,s}-\alpha^{\beta}_{0,s}} (A)
        }
        .
    \end{equation}

    To bound the first term in~\eqref{eq:proof-almost-spectral-flow-triangle-inequality-with-alpha-beta-X}, we use~\eqref{eq:comparing-automorphisms} to obtain
    \begin{equation}
        \label{eq:proof-almost-spectral-flow-alpha-alpha-beta-X-step-1}
        \begin{aligned}
            \Alignindent
            \abs[\Big]{
                \trace[\Big]{
                    P(0) \paren[\big]{\alpha_{0,s}-\alpha^{\beta,X}_{0,s}} (A)
                }
            }
            \\&\leq
            s \, \sup_{t\in\intervalcc{0,s}}
            \, \sumstack[lr]{Z\subset \Lambda}
            \, \abs[\Big]{
                \trace[\Big]{
                    P(t) \, \commutator[\Big]{
                        \calI_{H(t)}\paren[\big]{\dot \Phi(Z,t)} - \calI_{H(t),\beta_{X,Z}}\paren[\big]{\dot \Phi(Z,t)}
                        ,
                        \alpha^{\beta,X}_{t,s}(A)
                    }
                }
            }
            .
        \end{aligned}
    \end{equation}
    Now, for all projections \(P\) and operators \(V\) and \(W\)
    \begin{equation*}
        \abs[\big]{
            \trace[\big]{
                P
                \, \commutator{
                    V,W
                }
            }
        }
        =
        \abs[\big]{
            \trace[\big]{
                \commutator{
                    P,V
                }
                \, W
            }
        }
        \leq
        \norm{
            \commutator{
                P,V
            }
        }_1
        \, \norm{W}
        ,
    \end{equation*}
    by cyclicity of the trace and since the Hilbert space is finite-dimensional.
    Applying this and \cref{lem:bound-almost-inverse-Liouvillian-and-inverse-Liouvillian} to each summand in~\eqref{eq:proof-almost-spectral-flow-alpha-alpha-beta-X-step-1} gives
    \begin{equation*}
        \abs[\Big]{
            \trace[\Big]{
                P(0) \paren[\big]{\alpha_{0,s}-\alpha^{\beta,X}_{0,s}} (A)
            }
        }
        \leq
        2
        \, s
        \, \norm{P}_1
        \, \norm{A}
        \, \gamma^{-1}
        \, \sup_{t\in\intervalcc{0,s}}
        \, \sumstack[lr]{Z\subset \Lambda}
        \, \norm{\dot\Phi(Z,t)}
        \, \Exp{-\frac{\gamma^2}{4\beta_{X,Z}^2}}
        .
    \end{equation*}
    It remains to control the sum.
    With~\eqref{eq:definition-beta-X-Z}, we can factor out \(\e^{-\frac{\gamma^2}{4\beta^2}}\) and organize the sum as
    \begin{equation*}
        \sum_{Z\subset\Lambda}
        \norm{\dot\Phi(Z,t)}
        \, \Exp{-\frac{\gamma^2}{4}\unit_{\dist{X,Z} \geq \ell} \, \dist{X,Z}}
        =
        \sum_{n=0}^\infty
        \sumstack[r]{Z\subset\Lambda:\\\dist{Z,X}=n}
        \, \norm{\dot\Phi(Z,t)}
        \, \Exp{-\frac{\gamma^2}{4}\unit_{n \geq \ell}n}
        .
    \end{equation*}
    Since
    \begin{equation*}
        \sumstack[r]{Z\subset\Lambda:\\\dist{Z,X}=n}
        \norm{\dot\Phi(Z,t)}
        \leq
        \sum_{x\in X}
        \sumstack{z\in\Lambda:\\\dist{z,x}=n}
        \sumstack{Z\subset \Lambda:\\ z \in Z}
        \norm{\dot\Phi(Z,t)}
        \leq
        \abs{X} \, \abs{B_x(n)} \, \norm{\dot\Phi}_0
        ,
    \end{equation*}
    we conclude that
    \begin{align*}
        \sum_{Z\subset\Lambda}
        \norm{\dot\Phi(Z,t)}
        \, \Exp{-\frac{\gamma^2}{4}\unit_{\dist{X,Z} \geq \ell} \dist{X,Z}}
        &\leq
        \abs{X} \, \norm{\dot\Phi}_0 \, \Cvol \, \paren[\Big]{
            \sum_{n=0}^{\floor{\ell}} (n+1)^D
            + \sum_{n=\floor{\ell}+1}^{\infty} (n+1)^D \, \e^{-\frac{\gamma^2}{4}n}
        }
        .
    \end{align*}
    The first sum is bounded by \(C \, \ell^{D+1}\).
    For the second one, we use~\eqref{eq:bound-polynomial-times-exponential} and~\eqref{eq:bound-sum-exponential} to conclude that
    \begin{equation}
        \label{eq:truncated exponential sum}
        \sum_{n=\floor{\ell}+1}^{\infty} (n+1)^D \e^{-\frac{\gamma^2}{4}n}
        \leq
        C \, \frac{\e^{\frac{\gamma^2}{8}}}{\gamma^{2(D+1)}} \, \e^{-\frac{\gamma^2}{8} \ell}
        \leq
        C \, \e^{-c \ell}
        .
    \end{equation}
    Altogether,
    \begin{equation}
        \label{eq:easy}
        \abs[\Big]{
            \trace[\Big]{
                P(0) \, \paren[\big]{\alpha_{0,s}-\alpha^{\beta,X}_{0,s}} (A)
            }
        }
        \leq
        C
        \, \norm{P}_1
        \, \norm{A}
        \, \abs{X}
        \, \norm{\dot\Phi}_0
        \, \paren[\big]{\ell^{D+1} + \e^{-c\ell}}
        \, \e^{-\frac{\gamma^2}{8\beta^2}}
        .
    \end{equation}

    To bound the second term of~\eqref{eq:proof-almost-spectral-flow-triangle-inequality-with-alpha-beta-X}, we use~\eqref{eq:comparing-automorphisms} again to obtain
    \begin{equation*}
        \begin{aligned}
            \Alignindent
            \abs[\Big]{
                \trace[\Big]{
                    P(0) \, \paren[\big]{\alpha^{\beta,X}_{0,s}-\alpha^{\beta}_{0,s}} (A)
                }
            }
            \\&\leq
            s \, \sup_{t\in\intervalcc{0,s}}
            \, \sumstack[lr]{Z\subset \Lambda\suchthat\\\dist{X,Z} \geq \ell}
            \, \abs[\Big]{
                \trace[\Big]{
                    P(t) \, \commutator[\Big]{
                        \calI_{H(t),\beta_{X,Z}}\paren[\big]{\dot \Phi(Z,t)}
                        - \calI_{H(t),\beta}\paren[\big]{\dot \Phi(Z,t)}
                        ,
                        \alpha^{\beta}_{t,s}(A)
                    }
                }
            }
            ,
        \end{aligned}
    \end{equation*}
    where the restriction is due to the fact that otherwise \(\beta_{X,Z} = \beta\) and the terms cancel exactly.
    To obtain a bound for this term, we only use the locality of both flows.
    Denoting \(r=\dist{X,Z}\), we estimate the commutator as
    \begin{subequations}
        \begin{align}
            \Alignindent
            \nonumber
            \norm[\Big]{
                \commutator[\Big]{
                    \calI_{H(t),\beta_{X,Z}}\paren[\big]{\dot \Phi(Z,t)}
                    - \calI_{H(t),\beta}\paren[\big]{\dot \Phi(Z,t)}
                    ,
                    \alpha^{\beta}_{t,s}(A)
                }
            }
            \\\leq{} &
            \label{eq:proof-automorphic-equivalence-beta-X--beta-term-1}
            4 \, \norm[\big]{
                \dot \Phi(Z,t)
            }
            \, \norm[\big]{
                \paren[\big]{\id-\cexp{X_{r/2}}{}}\paren[\big]{\alpha^\beta_{t,s}(A)}
            }
            \\&
            \label{eq:proof-automorphic-equivalence-beta-X--beta-term-2}
            + \norm[\big]{
                \commutator[\big]{
                    \calI_{H(t),\beta}\paren[\big]{\dot \Phi(Z,t)}
                    ,
                    \cexp{X_{r/2}}{\alpha^\beta_{t,s}(A)}
                }
            }
            \\&
            \label{eq:proof-automorphic-equivalence-beta-X--beta-term-3}
            + \norm[\big]{
                \commutator[\big]{
                    \calI_{H(t),\beta_{X,Z}}\paren[\big]{\dot \Phi(Z,t)}
                    ,
                    \cexp{X_{r/2}}{\alpha^\beta_{t,s}(A)}
                }
            }
            ,
        \end{align}
    \end{subequations}
    where \(\cexp{X}{}\) is the conditional expectation as in the proof of \cref{lem:interaction-for-generator-of-almost-spectral-flow}.
    We shall use repeatedly below that in applications of \cref{lem:locality-inv-Liouvillian-beta,lem:interaction-for-generator-of-almost-spectral-flow}
    \begin{equation*}
        b(\beta) = \frac{\beta^2}{4v^2}
    \end{equation*}
    and that \(b(\beta)\) and \(\beta\) are bounded because of~\eqref{eq:Hyp on beta}.

    The first term~\eqref{eq:proof-automorphic-equivalence-beta-X--beta-term-1} is bounded using locality of \(\alpha^\beta_{t,s}\), \cref{lem:interaction-for-generator-of-almost-spectral-flow}:
    The flow \(\alpha^\beta\) is generated by an interaction \(\Psi_\beta\) such that \(
        \norm{\Psi_\beta}_{b(\beta)/3}
        \leq
        \paren[\big]{1 + C \, b(\beta)^{-(D+1)} \, (\beta + \beta^{-1})} \, \norm{\dot \Phi}_{b'}
    \).
    We then use \cref{lem:LRB} with \(b'\to \frac{b(\beta)}{3}\) and \(b\to \frac{2b(\beta)}{9}\) to obtain for all \(\abs{t-s} \leq 1\)
    \begin{equation}
        \label{eq:first term}
        4 \, \norm[\big]{
            \dot \Phi(Z,t)
        }
        \, \norm[\big]{
            \paren[\big]{\id-\cexp{X_{r/2}}{}}
            \, \alpha^\beta_{t,s}(A)
        }
        \leq
        C_1(\beta) \, \norm{\dot \Phi(Z,t)} \, \norm{A} \, \abs{X} \, \e^{-\frac{b(\beta) }{9}r}
        ,
    \end{equation}
    with
    \begin{equation*}
        C_1(\beta)
        \leq
        8 \, \Exp{2\Ccod{1,b(\beta)/9} \, \norm{\Psi_\beta}_{b(\beta)/3}}
        \leq
        8 \, \Exp{c \, b(\beta)^{-D}(1 + b(\beta)^{-(D+1)}(\beta + \beta^{-1}))}
        \leq
        C \, \e^{c \, \beta^{-(4D+3)}}
        ,
    \end{equation*}
    for some \(C>0\), where we used the bound~\(
        \Ccod{b(\beta)/9,1}
        \leq
        C \, b(\beta)^{-D}
    \), see~\eqref{eq:Cvol}.

    To bound the second term~\eqref{eq:proof-automorphic-equivalence-beta-X--beta-term-2} we use \cref{lem:locality-inv-Liouvillian-beta}.
    Specifically, \cref{eq:alternative locality} yields
    \begin{equation}
        \label{eq:second term}
        \norm[\big]{
            \commutator[\big]{
                \calI_{H(t),\beta}\paren[\big]{\dot \Phi(Z,t)}
                ,
                \cexp{X_{r/2}}{\alpha^\beta_{t,s}(A)}
            }
        }
        \leq
        C \, \beta^{-1} \, \norm{\dot \Phi(Z,t)} \, \norm{A} \, \abs{X}
        \, \e^{-\frac{b(\beta) }{2} r}
        ,
    \end{equation}
    for some \(C>0\).

    To control the last term~\cref{eq:proof-automorphic-equivalence-beta-X--beta-term-3}, we use the same \cref{lem:locality-inv-Liouvillian-beta}, this time in the form~\eqref{eq:lem-locality-inv-Liouvillian-beta-inf}, because \(\beta\) is replaced by \(\beta_{X,Z}\).
    We shall choose again \(T = \frac{r}{2v}\).
    With \(\beta_{X,Z} \leq \beta\), the first term of~\eqref{eq:lem-locality-inv-Liouvillian-beta-inf} is bounded by \(\beta \, \e^{-\frac{b}{2}r}\), up to the prefactors.
    For the second term, we use the Gaussian decay.
    For \(r\geq \ell\), we have that
    \begin{equation*}
        \beta_{X,Z}^2 \, T^2
        =
        \frac{1}{\beta^{-2}+r}\frac{r^2}{(2v)^2}
        \geq
        \frac{\beta^2}{1+r} \, \frac{r^2}{(2v)^2}
        \geq
        \frac{b(\beta) \, r}{2}
    \end{equation*}
    and in turn
    \begin{equation*}
        \frac{1}{\beta_{X,Z}} \, \e^{-\beta_{X,Z}^2 T^2}
        \leq
        \frac{1}{\beta} \, \sqrt{1+r}
        \, \e^{-\frac{b(\beta)}{2} r}
        ,
    \end{equation*}
    where we used \(\beta \leq 1\) and \(r \geq 1\) repetitively.
    We then conclude that
    \begin{equation*}
        \norm[\big]{
            \commutator[\big]{
                \calI_{H(t),\beta_{X,Z}}\paren[\big]{\dot \Phi(Z,t)}
                ,
                \cexp[\big]{X_{r/2}}{\alpha^\beta_{t,s}(A)}
            }
        }
        \leq
        C \, \abs{X} \, \norm{\dot \Phi(Z,t)} \, \norm{A}
        \, \paren[\Big]{
            \e^{-\frac{b}{2}r}
            + \beta^{-1} \, \sqrt{r} \, \e^{-\frac{b(\beta)}{2} r}
        }
        .
    \end{equation*}

    Combining the three terms and handling the sum \(Z\subset\Lambda\) with \(d(X,Z)\geq\ell\) as in the bound for~\eqref{eq:proof-almost-spectral-flow-alpha-alpha-beta-X-step-1}, we obtain
    \begin{equation*}
        \begin{aligned}
            \Alignindent
            \abs[\Big]{
                \trace[\Big]{
                    P(0) \, \paren[\big]{\alpha^{\beta,X}_{0,s}-\alpha^{\beta}_{0,s}} (A)
                }
            }
            \\&\leq
            C
            \, \norm{P}_1
            \, \norm{A}
            \, \abs{X}^2
            \, C\supint
            \, \sum_{n=\ell}^\infty
            \paren[\Big]{
                \e^{c \, \beta^{-(4D+3)}}
                \, \e^{-\frac{b(\beta)}{9}n}
                + \beta^{-1} \, \paren[\big]{1+\sqrt{n}} \, \e^{-\frac{b(\beta)}{2}n}
                + \e^{-\frac{b}{2}n}
            }
            \\&\leq
            C
            \, \norm{P}_1
            \, \norm{A}
            \, \abs{X}^2
            \, \paren[\Big]{
                \e^{c \, \beta^{-(4D+3)}}
                \, \beta^{-2}
                \, \e^{-\frac{b(\beta)}{9} \ell}
                + \beta^{-4} \, \e^{-\frac{b(\beta)}{4}\ell}
                + \e^{-\frac{b}{2}\ell}
            }
            ,
        \end{aligned}
    \end{equation*}
    by using~\eqref{eq:bound-polynomial-times-exponential}, \eqref{eq:bound-sum-exponential} and the properties of \(b(\beta)\) and \(\beta\).
    We now choose \(\ell = \tfrac{9}{b(\beta)} \, \paren{\beta^{-2} + c \, \beta^{-(4D+3)}}\) and absorb the polynomial dependence on \(\beta^{-1}\) in the exponential to obtain the upper bound
    \begin{equation*}
        C
        \, \norm{P}_1
        \, \norm{A}
        \, \abs{X}^2
        \, \e^{-c \, \beta^{-2}}
        .
    \end{equation*}
    Plugging this choice of \(\ell\) into~\eqref{eq:easy} and similarly absorbing the polynomial dependence, we obtain the same upper bound.
\end{proof}

It remains to prove LPPL\@.

\begin{proof}[Proof of \cref{thm:LPPL}]
    We first use triangle inequality
    \begin{equation*}
        \abs[\big]{
            \omega_s(A)
            - \omega_0(A)
        }
        \leq
        \abs[\big]{
            \omega_s(A)
            - \omega_0 \circ \alpha^\beta_{0,s}(A)
        }
        + \abs[\big]{
            \omega_0 \circ \alpha^\beta_{0,s}(A)
            - \omega_0 (A)
        }
        .
    \end{equation*}
    To bound the first summand, we again write
    \begin{equation*}
        \begin{aligned}
            \Alignindent
            \abs[\Big]{
                \trace[\Big]{
                    P(0) \, \paren[\big]{\alpha_{0,s}-\alpha^{\beta}_{0,s}} (A)
                }
            }
            \\&\leq
            s \, \sup_{t\in\intervalcc{0,s}}
            \, \sumstack[r]{Z\subset \Lambda\suchthat\\\dist{X,Z} \geq \dist{X,\Lambdap}}
            \, \abs[\Big]{
                \trace[\Big]{
                    P(t) \, \commutator[\Big]{
                        \calI_{H(t)}\paren[\big]{\dot \Phi(Z,t)}
                        - \calI_{H(t),\beta}\paren[\big]{\dot \Phi(Z,t)}
                        ,
                        \alpha^{\beta}_{t,s}(A)
                    }
                }
            }
            ,
        \end{aligned}
    \end{equation*}
    since \(\norm{\dot \Phi(Z,t)} = 0\) if \(\dist{X,Z} < \dist{X,\Lambdap}\) by assumption.
    Following the proof of \cref{thm:almost-spectral-flow}, we obtain
    \begin{equation*}
        \abs[\big]{
            \omega_s(A)
            - \omega_0 \circ \alpha^\beta_{0,s}(A)
        }
        \leq
        C
        \, \abs{X}^2
        \, \norm{A}
        \, \e^{-c \dist{X,\Lambdap}}
        .
    \end{equation*}
    For the second part, we use that \(\alpha^\beta\) acts almost trivially away from \(\Lambdap\).
    For this, we consider another automorphism \(\alphat^\beta\) generated by
    \begin{equation*}
        \Psit_\beta(Z,t)
        =
        \begin{cases}
            \Psi_\beta(Z,t) & \text{if }Z \cap X = \emptyset, \\
            0 & \text{otherwise,}
        \end{cases}
    \end{equation*}
    which is such that \(\alphat^\beta_{t,s}(A) = A\).
    Moreover, \(\Psi_\beta(Z,t)\), \(\Psit_\beta(Z,t) = 0\) if \(Z\cap \Lambdap = \emptyset\) by construction of \(\Psi_\beta\) in the proof of \cref{lem:interaction-for-generator-of-almost-spectral-flow} and the assumption that \(\dot \Phi(Z,t) = 0\) unless, \(Z \subset \Lambdap\).
    Then,
    \begin{align*}
        \abs[\big]{
            \omega_0 \circ \alpha^\beta_{0,s}(A)
            - \omega_0 (A)
        }
        &\leq
        \norm{P(0)}_1
        \, \norm[\big]{
            \alpha^\beta_{0,s}(A)
            -\alphat^\beta_{0,s}(A)
        }
        \\&\leq
        s
        \, \norm{P(0)}_1
        \,
        \sup_t
        \, \sumstack[lr]{Z\subset \Lambda}
        \, \norm[\big]{
            \commutator[\big]{
                \Psi(Z,t) - \Psit(Z,t)
                , \alphat^\beta_{t,s}(A)
            }
        }
        \\&\leq
        2
        \, s
        \, \norm{P(0)}_1
        \, \norm{A}
        \, \sup_t
        \, \sumstack[l]{x\in X}
        \, \sumstack{z\in \Lambdap}
        \, \sumstack[r]{Z\subset \Lambda\suchthat\\x,z\in Z}
        \, \norm[\big]{\Psi(Z,t)}
        \\&\leq
        2
        \, s
        \, \norm{P(0)}_1
        \, \norm{A}
        \, \abs{X}
        \, \norm{\Psi}_{b(\beta)/3}
        \, \sup_{x\in \Lambda}
        \, \sumstack{z\in \Lambdap}
        \, \e^{-b(\beta) \dist{x,z} /3}
        .
    \end{align*}
    Bounding this sum and combining both bounds gives the result.
\end{proof}
Note that the choice of Gaussian width \(\beta\) could be optimized to give the sharpest decay rate (following the above, \(\beta\) should be of order \(d(X,\Lambdap)^{-\frac{1}{2}}\)).

\section{Exponential clustering revisited}
\label{sec:clustering}

Finally, Gaussian filters were already used in previous proofs of the exponential clustering theorem: Ground state correlations decay in the presence of a spectral gap.
In the present context of quantum spin systems, exponential decay of correlations were proved in~\cite{NS2006} in the infinite-volume limit (the gap refers in this case to the gap of the GNS Hamiltonian), in the limit of finite volumes in~\cite{HK2006}, where the splitting of the ground state energies in finite volume is assumed to vanish in the limit, while the latter assumption in removed in~\cite{BBRF2020rational} but the decay is only superpolynomial.
In this section, we prove that correlations decay indeed exponentially under the sole assumption of a spectral gap, even if there is eigenvalue splitting in the ground state.

\begin{assumption}
    \label{assumption:gap-and-width}
    Let \(\Lambda\) be finite.
    Let \(\Phi\) be an interaction satisfying \(\norm{\Phi}_{b'}<\infty\) for some \(b'>0\), and let \(H\) be the corresponding Hamiltonian.
    We assume that the spectrum of \(H\) is of the form
    \begin{equation*}
        \sigma\paren{H} = \sigma_0 \cup \sigma_1
    \end{equation*}
    with \(\inf(\sigma_1)-\sup(\sigma_0) \geq \gamma > 0\) and \(\diam{\sigma_0} \leq \Delta < \gamma/4\).
\end{assumption}

As above, we denote by \(P\) the spectral projection associated with the spectral patch \(\sigma_0\).
Note that the condition \(\Delta<\gamma/4\) is not tight but will simplify the estimates.
What is needed in the proof is that \(-\Delta+\gamma/2 \geq c > 0\).

\begin{theorem}
    \label{thm:exponential clustering}
    Let \(D\in \N\), \(\Cvol > 0\), \(b'>0\), \(C\supint>0\), \(\gamma>0\) and \(\Delta>0\).
    Then there exist constants \(C\), \(c > 0\), such that the following holds.
    For all \(\Lambda\in \graphs\) finite and Hamiltonians \(H\) that satisfy \cref{assumption:gap-and-width} with gap \(\gamma\) and width \(\Delta\) and are given by interactions \(\Phi\) such that \(\norm{\Phi}_{b'}<C\supint\) the following holds:

    For any normalized state \(\Omega \in \mathrm{Ran}(P)\),
    \begin{equation*}
        \abs{
            \braket{\Omega,A \, B \, \Omega}
            - \braket{\Omega,A \, P \, B \, \Omega}
        }
        \leq
        C \, \norm{P}_1 \norm{A} \, \norm{B} \, \e^{-c d(X,Y)}
    \end{equation*}
    for all disjoint \(X\), \(Y\subset \Lambda\) and \(A\in\alg_X\), \(B\in\alg_Y\).
\end{theorem}

\begin{proof}
    For the proof, we define the filter function
    \begin{equation*}
        g_\beta(t)
        =
        \frac{\Exp{-\I t\frac{\gamma}{2}}}{\sqrt{2}\beta}
        \, \paren[\bigg]{
            \sqrt{\frac{\pi}{2}} \, \delta_0 (t)
            + \frac{\I}{\sqrt{2\pi}} \, \pv[\Big]{\frac{1}{t}}
        } \, \phi_\beta(t)
        ,
    \end{equation*}
    where \(\pv[\big]{\frac{1}{t}}\) denotes the principal value distribution, and let
    \begin{equation*}
        \calJ_{H,\beta}(A)
        =
        \frac{1}{\sqrt{2\pi}}
        \, \int_{-\infty}^\infty \odif{t} \, g_\beta(t) \, \tau_t^{H}(A)
        .
    \end{equation*}
    The Fourier transform of \(\hat{g}_\beta\) is given by
    \begin{equation*}
        \hat{g}_\beta(\omega)
        =
        \frac{1}{\sqrt{2}\beta}
        \, \paren{\Theta_{\gamma/2} \star \hat \phi_\beta}(\omega)
        =
        \frac{1}{2\beta\sqrt {\pi}}
        \, \int_{-\infty}^{\omega-\gamma/2} \odif{\xi} \, \Exp{-\frac{\xi^2}{4\beta^2}}
        =
        \frac{1}{\sqrt{\pi}}
        \, \int_{-\infty}^{\frac{\omega-\gamma/2}{2\beta}} \odif{x} \, \Exp{-x^2}
        ,
    \end{equation*}
    where \(\Theta_a\) is the Heaviside step function with discontinuity at \(a\).

    Before we continue the proof, we note some fact about the Gaussian error function.
    We have
    \begin{equation*}
        1-\frac{1}{\sqrt{\pi}} \, \int_{-\infty}^z \odif{x} \, \Exp{-x^2}
        =
        \frac{1}{\sqrt{\pi}} \, \int_z^{\infty} \odif{x} \, \Exp{-x^2}
        \leq
        \frac{1}{2 \, \sqrt{\pi} \, z} \, \int_z^\infty \odif{x} \, 2 \, x \, \Exp{-x^2}
        =
        \frac{1}{2 \, \sqrt{\pi}} \, \frac{\Exp{-z^2}}{z}
    \end{equation*}
    whenever \(z>0\), and similarly
    \begin{equation*}
        \frac{1}{\sqrt{\pi}} \int_{-\infty}^z \odif{x} \, \Exp{-x^2}
        \leq
        \frac{1}{2\,\sqrt{\pi}} \, \frac{\Exp{-z^2}}{\abs{z}}
    \end{equation*}
    if \(z<0\).
    Hence, the Fourier transform of the filter function satisfies
    \begin{align*}
        \hat{g}_\beta(\omega)
        &\leq
        \frac{\beta}{\sqrt{\pi} \, \abs{\omega - \gamma/2} }
        \, \Exp{-\frac{(\omega-\gamma/2)^2}{4\beta^2}}
        \qquad\text{for \(\omega<\gamma/2\)}
    \shortintertext{and}
        1- \hat{g}_\beta(\omega)
        &\leq
        \frac{\beta}{\sqrt{\pi} \, (\omega - \gamma/2)}
        \, \Exp{-\frac{(\omega-\gamma/2)^2}{4\beta^2}}
        \qquad\text{for \(\omega>\gamma/2\).}
    \end{align*}

    We now decompose the correlation into three terms
    \begin{subequations}
        \begin{align}
            \braket{\Omega, A \, B \Omega} - \braket{\Omega, A \, P \, B \, \Omega}
            ={}&
            \braket{\Omega, A \, P^\perp \, B \, \Omega} \nonumber
            \\={}&
            \braket[\big]{
                \Omega
                ,
                \commutator{\calJ_{H,\beta}(A),B} \, \Omega
            }
            \label{I}
            \\&+
            \braket[\big]{
                \Omega
                ,
                B \, \calJ_{H,\beta}(A) \, \Omega}
            \label{II}
            \\&+
            \braket[\big]{
                \Omega,
                \paren[\big]{P \, A \, P^\perp - P \, \calJ_{H,\beta}(A)} \, B \, P \, \Omega}
            \label{III}
            .
        \end{align}
    \end{subequations}
    In the following we will show that each term decays exponentially in \(d(X,Y)\) with the choice \(\beta = \frac{\gamma}{2 \, \sqrt{d(X,Y)}}\).

    We start by bounding~\eqref{II}, for which we observe
    \begin{equation*}
        \calJ_{H,\beta}(A) \, P
        =
        \sumstack[l]{\mu\in \sigma(H)}
        \sum_{\nu\in \sigma_0}
        \hat{g}_\beta(\nu-\mu)
        \, P_\mu \, A \, P_\nu
        .
    \end{equation*}
    Since \(\nu-\mu \leq \Delta < \gamma/4\), we obtain the bound
    \begin{equation}
        \label{II bounded}
        \norm{
            \calJ_{H,\beta}(A) \, P
        }
        =
        \norm{A}
        \, \sum_{\nu\in \sigma_0}
        \norm[\Big]{
            \sum_{\mu\in \sigma(H)}
            \hat{g}_\beta(\nu-\mu)
            \, P_\mu
        }
        \leq
        \frac{4 \, \abs{\sigma_0}}{\sqrt{\pi}}
        \, \frac{\beta}{\gamma}
        \, \Exp{-\frac{1}{64}\paren*{\frac{\gamma}{\beta}}^2}
        \, \norm{A}
        .
    \end{equation}
    For~\eqref{III} we bound
    \begin{equation*}
        \norm{
            P \, A \, P^\perp - P \, \calJ_{H,\beta}(A) \, P^\perp
        }
        =
        \norm[\bigg]{
            \sum_{\mu\in \sigma_0}
            \sum_{\nu\in \sigma_1}
            \paren[\big]{1-\hat{g}_\beta(\nu-\mu)}
            \, P_\mu \, A \, P_\nu
        }
        \leq
        \frac{2 \, \abs{\sigma_0}}{\sqrt{\pi}}
        \, \frac{\beta}{\gamma}
        \, \Exp{-\frac{1}{16}\paren*{\frac{\gamma}{\beta}}^2}
        \, \norm{A}
    \end{equation*}
    because \(\nu-\mu \geq \gamma\).
    And together with~\eqref{II bounded}, we obtain
    \begin{equation}
        \label{III bounded}
        \norm{P \, A \, P^\perp - P \, \calJ_{H,\beta}(A)}
        \leq
        \frac{6 \, \abs{\sigma_0}}{\sqrt{\pi}}
        \, \frac{\beta}{\gamma}
        \, \Exp{-\frac{1}{64}\paren*{\frac{\gamma}{\beta}}^2}
        \, \norm{A}
        .
    \end{equation}

    Finally, the commutator \(\commutator{\calJ_{H,\beta}(A),B}\) in~\eqref{I} is bounded using an argument similar to that of Lemma~\ref{lem:locality-inv-Liouvillian-beta}.
    As we do there, we decompose the integral defining \(\calJ_{H,\beta}(A)\) into \(\abs{t}\leq T\) and \(\abs{t}>T\) and use the Lieb-Robinson bound, Lemma~\ref{lem:LRB}, to estimate the short time part while the long time contribution can be bounded by the Gaussian decay.
    Since \(d(X,Y)>0\), we have that \(\commutator{A,B} = 0\) and so the \(\delta_0\)-contribution vanishes.
    Therefore,
    \begin{align*}
        \Alignindent
        \norm[\bigg]{
            \int_{\abs{t}< T} \odif{t} \, g_{\beta}(t) \, \commutator{\tau^H_t(A),B}
        }
        \\&\leq
        \frac{2 \, \Ccod{1,b'-b}^{-1}}{\beta \, \sqrt{\pi}}
        \, \norm{A} \, \norm{B}
        \, \min \List[\big]{\abs{X},\abs{Y}}
        \, \sup_{t} \abs{\phi_\beta (t)}
        \, \e^{-b\dist{X,Y}}
        \, \lim_{\epsilon \to 0}
        \, \int_{\epsilon}^{T} \odif{t} \, \frac{1}{t} \, \paren[\big]{\e^{bvt}-1}
        .
    \end{align*}
    The mean value theorem implies that \(t^{-1} \, (\e^{bvt}-1) \leq b \, v \, \e^{bvt}\) and so
    \begin{equation*}
        \lim_{\epsilon \to 0} \int_{\epsilon}^{T} \odif{t}
        \, \frac{1}{t} \, \paren[\big]{\e^{bvt}-1}
        \leq
        \paren[\big]{\e^{bvT}-1}
        .
    \end{equation*}
    Since \(\sup \abs{\phi_\beta(t)} = \frac{\beta}{\sqrt{\pi}}\), we conclude that
    \begin{equation}
        \label{I.1 bounded}
        \norm[\bigg]{
            \int_{\abs{t}< T} \odif{t} \, g_{\beta}(t) \, \commutator{\tau^H_t(A),B}
        }
        \leq
        \frac{2 \, \Ccod{1,b'-b}^{-1} \, \norm{A} \, \norm{B}}{\pi}
        \, \min \List[\big]{\abs{X},\abs{Y}}
        \, \e^{-b\dist{X,Y}} \, \paren[\big]{\e^{bvT}-1}
        .
    \end{equation}
    For \(\abs{t} \geq T\), we use the simple norm bound on the commutator and
    \begin{equation*}
        \int_T^\infty \odif{t} \, \frac{\phi_\beta(t)}{t}
        \leq
        \frac{1}{T} \, \int_T^\infty \odif{t} \, \phi_\beta(t)
        \leq
        \frac{1}{2 \, T^2 \, \sqrt{\pi}} \, \e^{-\beta^2 T^2}
    \end{equation*}
    to conclude that
    \begin{equation}
        \label{I.2 bounded}
        \norm[\bigg]{
            \int_{\abs{t} > T} \odif{t} \, g_{\beta}(t) \, \commutator{\tau^H_t(A),B}
        }
        \leq
        \frac{2 \, \norm{A} \, \norm{B}}{T^2 \, \sqrt{\pi}} \, \e^{-\beta^2 T^2}
        .
    \end{equation}
    Together, \eqref{I.1 bounded}, \eqref{I.2 bounded} and the choice \(T=\frac{d(X,Y)}{2v}\) yield the following bound on~\eqref{I}:
    \begin{equation}
        \label{I bounded}
        \norm[\bigg]{
            \int_{-\infty}^\infty \odif{t} \, g_{\beta}(t) \, \commutator{\tau^H_t(A),B}
        }
        \leq
        \frac{2 \, \norm{A} \, \norm{B}}{\pi}
        \, \paren[\bigg]{
            \frac{\min \List[\big]{\abs{X},\abs{Y}}}{\Ccod{1,b'-b}}
            \, \e^{-\frac{b d(X,Y)}{2}}
            + \frac{4 \, v^2\, \sqrt{\pi}}{d(X,Y)^2}
            \, \e^{-\frac{\beta^2 d(X,Y)^2}{4v^2}}
        }
        .
    \end{equation}
    Gathering~\eqref{I bounded}, \eqref{II bounded} and~\eqref{III bounded} to bound~\eqref{I}, \eqref{II} and\eqref{III}, the claim of the theorem follows by bounding \(\abs{\sigma_0} \leq \norm{P}_1\) and setting
    \begin{equation*}
        \beta = \frac{\gamma}{2 \, \sqrt{d(X,Y)}}
        ,
    \end{equation*}
    which makes all exponents proportional to \(d(X,Y)\), and using that \(d(X,Y) \geq 1\) in the prefactors.
\end{proof}


\section{Putting it all together: the quantum Hall effect}
\label{sec:QHE}

We briefly recall the setting of~\cite{BBRF2020rational} applied to the quantum Hall effect.
The lattice is a sequence of discrete tori \(\Lambda_L = (\Z/L\Z)^2\).
The Hamiltonian is given by a finite range interaction \(\Phi\) that is invariant under a strictly local \(U(1)\)-action.
This means that there is a family \(q_x = q_x^*\in\alg_x\) with integer spectrum and such that \(\commutator{Q_{\Lambda_L},\Phi(Z)} = 0\) for all \(Z\subset\Lambda_L\).
The \(q_x\) are the \emph{local charges}.
The \enquote{ground state space} is the range of a spectral projection \(P_L\) of \(H_{\Lambda_L}\) whose dimension is constant and equal to \(p\) for all \(L\) large enough.
Moreover, \(H_{\Lambda_L}\) is assumed to satisfy the \nameref{assumption:gap}, uniformly in \(L\) for \(L\) large enough.

The proof of quantization of the Hall conductance in~\cite{BBRF2020rational} relies heavily on the inverse Liouvillian on the one hand, and on clustering on the other hand.
The inverse Liouvillian is used to construct a unitary \(U_L\) describing a magnetic flux threading and its locality allows for the definition of a charge transport operator \(T_L\) across a fiducial line of the torus.
Replacing the exact inverse Liouvillian by the almost inverse Liouvillian introduced in Section~\ref{sec:almost-inverse-Liouvillian} yields a unitary \(U_L^\beta\) and in turn an exponentially localized charge transport operator \(T^\beta_L\).
This improved localization and the exponential clustering of Section~\ref{sec:clustering} yield the following.

\begin{theorem}
    Let \(\beta = L^{-1/2}\).
    There is an integer \(n_L\in\Z\) and constants \(C\), \(c>0\) such that
    \begin{equation*}
        \abs[\big]{n_L - \trace[\big]{P_L \, T_L^\beta}}
        \leq C \, \e^{-c L}
    \end{equation*}
    for all~\(L\).
    If, moreover, the sequence of states \(p^{-1}\Tr(P_L(\cdot))\) is convergent, then \(n_L = n\) for \(L\) large enough and \(2 \, \pi \, \kappa = \frac{n}{p}\), where \(\kappa\) is the Hall conductance.
\end{theorem}

We now explain the arguments with a focus on the changes from using the almost inverse Liouvillian, referring to~\cite[Section~IV]{BBRF2020rational} for more details.
The geometric setting is described in~\cite[Section~2]{BRFL2021}, see in particular Figure~1 therein, where \(\eta_-\)~and~\(\nu_-\) correspond to what we will call the lower boundary of the upper half and the left boundary of the right half, respectively, of the torus.
We again drop the \(L\) dependence to simplify notations.

Similarly to~\cite[Section~2]{BRFL2021} we first construct a unitary \(U^\beta\) that models the threading of one unit of flux through the torus.
In our case it is exponentially localized near a line along the torus.
Therefore, let \(Q_\Uhalf\) be the charge on the upper half of the torus and let
\begin{equation*}
    \overline Q_\Uhalf^\beta = Q_\Uhalf-\calI_{H,\beta}\circ\calL_H(Q_\Uhalf)
    .
\end{equation*}
By the \(U(1)\)-invariance and the finite range condition of the Hamiltonian, the operator \(\calL_H(Q_\Uhalf)\) is strictly localized in two strips of finite width around the boundary of the half-torus, which we denote \(\calL_H(Q_\Uhalf)_\lowerboundary\) and \(\calL_H(Q_\Uhalf)_\upperboundary\), respectively.
The locality result, Lemma~\ref{lem:locality-inv-Liouvillian-beta}, implies that \(\calI_{H,\beta}(\calL_H(Q_\Uhalf)_\lowerboundary)\) is exponentially localized around the lower boundary.
Specifically, the choice \(\beta = L^{-1/2}\) in~\eqref{eq:alternative locality} yields a localization estimate of the form \(C \, \e^{-c L}\) for \(\dist{X,Y}\) of order \(L\), namely
\begin{equation*}
    \calI_{H,\beta}\circ\paren[\big]{\calL_H(Q_\Uhalf)_\lowerboundary}
\end{equation*}
can be approximated by an operator that is strictly localized in a strip of width \(\frac{L}{4}\) at the lower boundary (or any other smaller fraction of \(L\)), up to errors that are exponentially small in \(L\).
This and the fact that \(Q_\Uhalf\) has integer spectrum imply that the unitary \(\e^{2\pi\I\overline Q_\Uhalf^\beta}\) factorizes up to exponentially small errors into parts at the lower and upper boundary.
We let
\begin{equation}
    \label{eq:U beta}
    U^\beta = \paren[\big]{\e^{2\pi\I\overline Q_\Uhalf^\beta}}_\lowerboundary
\end{equation}
be the factor localized on the lower boundary of the half-torus.

We briefly pause the argument to compare explicitly with~\cite{BBRF2020rational}.
There, \(\overline Q_\Uhalf\) is defined using the exact inverse Liouvillian \(\calI_{H}\) rather than \(\calI_{H,\beta}\).
As a result, \(\commutator{\overline Q_\Uhalf,P}=0\) while here \(
    \norm{\commutator{\overline Q_\Uhalf^\beta , P}}
    \leq
    C \, L^2 \, \Exp{-c\beta^{-2}}
    =
    C \, \Exp{-c L}
\), by Proposition~\ref{prop:almost-inverse-Liouvillian-bound}.
The same holds for the exponentials, and by exponential clustering, Theorem~\ref{thm:exponential clustering}, for their restrictions to the lower boundary of the half-torus.
What is more, Lemma~\ref{lem:bound-almost-inverse-Liouvillian-and-inverse-Liouvillian} implies that \(\overline Q_\Uhalf^\beta\) and \(\overline Q_\Uhalf\) are exponentially close and therefore so are \(U^\beta\) and \(U\).
As a consequence, \(U^\beta\) almost preserves the ground state space and almost implements \(2\pi\)-flux threading.

Next, we consider the charge transport along the torus.
Therefore, let \(Q_\Rhalf\) be the charge on the right half of the torus.
By charge conservation and again the locality lemma, the operator \((U^\beta)^* \, Q_\Rhalf \, U^\beta - Q_\Rhalf\) decomposes in two contributions at each boundary of the right half-torus.
We define the operator of charge transport as the left one
\begin{equation}
    \label{eq:T}
    T^\beta
    =
    \paren*{(U^\beta)^* \, Q_\Rhalf \, U^\beta - Q_\Rhalf}_\leftboundary
    .
\end{equation}
More precisely, \(\trace[\big]{P_L \, T_L^\beta}\) measures the amount of charge transported in the ground state across one fiducial line across the torus after threading one unit of flux in the torus -- up to exponentially small errors in \(L\), because the splittings we used are only unique up to exponentially small terms.

With these definitions, the proof of quantization of this quantity follows exactly the original argument of~\cite[Section~IV.B]{BBRF2020rational}.
We only recall the main steps and illustrate where the results of the previous sections yield improved bounds.
The unitary
\begin{equation*}
    Z^\beta(\phi)
    =
    (U^\beta)^*
    \, \e^{\I\phi\overline Q_\Rhalf^\beta}
    \, U^\beta
    \, \e^{-\I\phi\overline Q_\Rhalf^\beta}
\end{equation*}
factorizes as \(Z^\beta(\phi) = Z^\beta(\phi)_\leftboundary \, Z^\beta(\phi)_\rightboundary\) as in the discussion above.
Note that the \enquote{\(-\)} and \enquote{\(+\)} in~\cite{BBRF2020rational} are the analogues of \enquote{\(\leftboundary\)} and \enquote{\(\rightboundary\)} here, respectively.
Since \(\calI_{H,\beta}\) is an almost inverse Liouvillian, Proposition~\ref{prop:almost-inverse-Liouvillian-bound}, \(Z^\beta(\phi)\) commutes with \(P\) (in operator norm) up to exponentially small errors in \(\beta^{-2} = L\).
By exponential clustering, Theorem~\ref{thm:exponential clustering}, it follows that \(\commutator{Z^\beta(\phi)_\leftboundary,P}\expeq 0\), where we use the notation \(\expeq\) to indicate equality up to exponentially small errors in \(L\).
From there on, the argument runs without a change, but the errors are always exponential rather than superpolynomial, yielding that
\begin{equation}
    \label{eq:Zbeta}
    P \, Z^\beta(\phi)_\leftboundary \, P
    \expeq
    \e^{\I\phi \paren[\Big]{P \paren[\big]{T^\beta - ((U^\beta)^* K^\beta_\leftboundary U^\beta - K^\beta_\leftboundary)}P + \overline Q_\Rhalf^\beta}}
    \, \e^{\I\phi \overline Q_\Rhalf^\beta}
    \, P
    ,
\end{equation}
where \(
    K^\beta_\leftboundary
    =
    \calI_{H,\beta}\circ \paren[\big]{\calL_H(Q_\Rhalf)_\leftboundary}
\).
At \(\phi = 2\pi\), an independent argument yields that \(\det_P Z^\beta(2\pi)_\leftboundary \expeq 1\), where the determinant is on the range of \(P\).
With this, the claim that \(\trace[\big]{P \, T^\beta}\) is exponentially close to an integer follows from~\eqref{eq:Zbeta} by computing the trace of the exponent and using the unitary invariance of the trace.

Finally, if the sequence of states is convergent, then \(p^{-1} \, \trace{P_L \, T_L}\) is convergent because \(T_L\) is a sufficiently local operator, see~\cite[Corollary~2.3]{bachmann2020many}.
Note that \(T_L\) is obtained with the exact inverse Liouvillian here.
The Laughlin argument then implies that the limit of \(p^{-1} \, \trace{P_L \, T_L}\) is equal to \(2\,\pi\,\kappa\) where \(\kappa\) is the Hall conductance, see~\cite[Theorem~3.2]{bachmann2020many}.
By Lemma~\ref{lem:bound-almost-inverse-Liouvillian-and-inverse-Liouvillian}, we further have that \(\norm{T^\beta_L - T_L} \to 0\) as \(L\to\infty\), again with \(\beta^{-2} = L\).
We conclude that \(p^{-1} \, \trace[\big]{P_L (T^\beta_L - T_L)}\) converges to \(0\), and hence that \(p^{-1} \, \trace{P_L \, T^\beta_L}\) is convergent.
In particular the sequence \(\paren[\big]{\frac{n_L}{p}}_{L}\) is eventually constant and equal to \(2\pi\kappa\).

\AtNextBibliography{\small}
\printbibliography[heading=bibintoc]

\end{document}

%% file: common.tex
\usepackage[utf8]{inputenc}
\usepackage[main=english]{babel}
\usepackage{libertinus}
\usepackage{csquotes}
\usepackage{enumitem}
\usepackage{microtype}

\usepackage{mathtools}
\mathtoolsset{centercolon,mathic}
\allowdisplaybreaks[2] 
\usepackage{amsthm}
\usepackage{thmtools}
\usepackage{derivative}
\usepackage{cancel}
\usepackage[libertinus,slantedGreek,varbb]{newtxmath}
\let\emptyset\varnothing
\usepackage[scr=boondoxo,cal=boondoxo]{mathalpha}

\let\phi\varphi

\let\epsilon\varepsilon
\let\epsi\varepsilon

\renewcommand{\pi}{\uppi}

\DeclarePairedDelimiter{\intervaloo}{\lparen}{\rparen}

\DeclarePairedDelimiter{\intervalcc}{\lbrack}{\rbrack}

\DeclarePairedDelimiter{\abs}{\lvert}{\rvert}
\DeclarePairedDelimiter{\paren}{\lparen}{\rparen}

\DeclarePairedDelimiter{\norm}{\lVert}{\rVert}

    \newcommand{\VERT}[1]{#1|\mkern-1.5mu#1|\mkern-1.5mu#1|}

    \ExplSyntaxOn
    \makeatletter

    \MT_delim_default_inner_wrappers:n{\tnorm}

    \NewDocumentCommand{\tnorm}{ s o m }{
        \IfBooleanTF{#1}{
        \MT_delim_tnorm_star_wrapper:nnn%
            {\VERT{\bgroup\left}}{#3}{\VERT{\aftergroup\egroup\right}}
        }{
            \IfValueTF{#2}{
                \@nameuse{MT_delim_tnorm_nostarscaled_wrapper:nnn}%
                    {\VERT{\@nameuse {\MH_cs_to_str:N #2 l}}}
                    {#3}
                    {\VERT{\@nameuse {\MH_cs_to_str:N #2 r}}}
            }{
                \MT_delim_tnorm_nostarnonscaled_wrapper:nnn%
                    {\VERT{}}
                    {#3}
                    {\VERT{}}
            }
        }
    }

    \makeatother
    \ExplSyntaxOff

\DeclarePairedDelimiter{\commutator}{\lbrack}{\rbrack}

\DeclarePairedDelimiter{\List}{\{}{\}}
\DeclarePairedDelimiter{\ceil}{\lceil}{\rceil}
\DeclarePairedDelimiter{\floor}{\lfloor}{\rfloor}

\DeclarePairedDelimiter{\braket}{\langle}{\rangle}

\DeclarePairedDelimiterXPP{\dd}[1]{d}{\lparen}{\rparen}{}{#1}
\DeclarePairedDelimiterXPP{\dist}[1]{d}{\lparen}{\rparen}{}{#1}
\DeclarePairedDelimiterXPP{\diam}[1]{\SYMdiam}{\lparen}{\rparen}{}{#1}
\DeclarePairedDelimiterXPP{\pv}[1]{\operatorname{p.v.}}{\lparen}{\rparen}{}{#1}

\makeatletter
\let\bBigg@@\bBigg@
\renewcommand{\bBigg@}[2]{{%
  \mathchoice
    {\bBigg@@{#1}{#2}}%
    {\bBigg@@{#1}{#2}}%
    {\big@size=.5\big@size\bBigg@@{#1}{#2}}%
    {\big@size=.3\big@size\bBigg@@{#1}{#2}}}}%
\makeatother

\DeclareDocumentCommand{\trace}{s o e{_} m}{
    \IfValueTF{#3}
        {\Tr_{#3}}
        {\Tr}%
    \IfBooleanTF{#1}
        {\paren*{#4}}
        {
            \IfValueTF{#2}
                {\paren[#2]{#4}}
                {\paren{#4}}%
        }%
}
\DeclareDocumentCommand{\Trace}{s o e{_} m}{
    \IfValueTF{#3}
        {\Tr_{#3}}
        {\Tr}%
    \IfBooleanTF{#1}
        {\paren*{#4}}
        {
            \IfValueTF{#2}
                {\paren[#2]{#4}}
                {\paren{#4}}%
        }%
}

\providecommand\given{}

\newcommand\SetSymbol[1][]{%
    \nonscript\,#1\vert
    \allowbreak
    \nonscript\,
    \mathopen{}}
\DeclarePairedDelimiterX\Set[1]\{\}{%
    \renewcommand\given{%
        \SetSymbol[\delimsize]}
    \nonscript\,
    #1
    \nonscript\,
}

\ExplSyntaxOn
\makeatletter
\cs_generate_variant:Nn \tl_if_eq:nnT {onT}
\tl_new:N \__scale_new_delimsize_tl
\cs_new:Nn \__scale_make_bigger_tl:N {
    \tl_if_eq:onT {#1} {\@MHempty} { \tl_set:Nn \__scale_new_delimsize_tl { big } }
    \tl_if_eq:onT {#1} {}          { \tl_set:Nn \__scale_new_delimsize_tl { big } }
    \tl_if_eq:onT {#1} {\big}      { \tl_set:Nn \__scale_new_delimsize_tl { Big } }
    \tl_if_eq:onT {#1} {\Big}      { \tl_set:Nn \__scale_new_delimsize_tl { bigg } }
    \tl_if_eq:onT {#1} {\bigg}     { \tl_set:Nn \__scale_new_delimsize_tl { Bigg } }
    \tl_if_eq:onT {#1} {\middle}   { \tl_set:Nn \__scale_new_delimsize_tl { middle } }
}
\cs_new:Nn \scale_make_bigger_m:N {
    \__scale_make_bigger_tl:N #1
    \use:c { \tl_use:N \__scale_new_delimsize_tl }
}
\cs_new:Nn \scale_make_bigger_l:N {
    \__scale_make_bigger_tl:N #1
    \tl_if_eq:NnTF \__scale_new_delimsize_tl {middle} {
        \tl_set:Nn \__scale_new_delimsize_tl {left}
    }{
        \tl_put_right:Nn \__scale_new_delimsize_tl {l}
    }
    \use:c { \tl_use:N \__scale_new_delimsize_tl }
}
\cs_new:Nn \scale_make_bigger_r:N {
    \__scale_make_bigger_tl:N #1
    \tl_if_eq:NnTF \__scale_new_delimsize_tl {middle} {
        \tl_set:Nn \__scale_new_delimsize_tl {right}
    }{
        \tl_put_right:Nn \__scale_new_delimsize_tl {r}
    }
    \use:c { \tl_use:N \__scale_new_delimsize_tl }
}
\DeclarePairedDelimiterXPP{\pd}[1]{\scale_make_bigger_l:N\delimsize\lparen d}{\lparen}{\rparen}{\scale_make_bigger_r:N\delimsize\rparen}{#1}
\DeclarePairedDelimiterXPP{\pdist}[1]{\scale_make_bigger_l:N\delimsize\lparen d}{\lparen}{\rparen}{\scale_make_bigger_r:N\delimsize\rparen}{#1}
\DeclarePairedDelimiterXPP{\pdiam}[1]{\scale_make_bigger_l:N\delimsize\lparen \SYMdiam}{\lparen}{\rparen}{\scale_make_bigger_r:N\delimsize\rparen}{#1}

\makeatother
\ExplSyntaxOff

\usepackage[svgnames]{xcolor}

\usepackage[
    colorlinks,
    allcolors=col_1,
    bookmarksnumbered,
    bookmarksopen,
    bookmarksopenlevel=2,
    bookmarksdepth=3,
    pdfusetitle=true,
    breaklinks=true,
]{hyperref}

\makeatletter
\@ifpackageloaded{hyperref}{
    \usepackage[capitalise,noabbrev]{cleveref}
    \usepackage{orcidlink}
}{
    \usepackage[capitalise,noabbrev,poorman]{cleveref}
    \newcommand{\texorpdfstring}[2]{#1}
    \newcommand{\href}[2]{#2}
    \newcommand{\nameref}[1]{name}
    \newcommand{\hypersetup}[1]{}
    \newcommand{\orcidlink}[1]{ORCiD}
    \newcommand{\pdfbookmark}[1]{}
}
\makeatother

\crefdefaultlabelformat{#2#1#3}
\crefname{equation}{}{}

\usepackage{caption}
\usepackage{subcaption}

\newcommand{\sumstack}[2][]{\ifstrempty{#1}{\sum_{\substack{#2}}}{\smashoperator[#1]{\sum_{\substack{#2}}}}}

\newcommand{\e}{{\mathrm{e}}}
\newcommand{\I}{\mathrm{i}}
 \newcommand{\R}{\mathbb{R}}
\newcommand{\C}{\mathbb{C}}
\newcommand{\N}{\mathbb{N}}
\newcommand{\Z}{\mathbb{Z}}
\newcommand{\Hi}{{\mathcal{H}}}
\newcommand{\alg}{\mathcal{A}}

\newcommand{\unit}{\mathbf{1}}
\newcommand{\Exp}[1]{\e^{#1}}
\DeclareMathOperator{\id}{id}
\newcommand\cexpsym{\mathbb{E}}
\ExplSyntaxOn
\DeclareDocumentCommand{\cexp}{s o m m}{%
    \cexpsym\c_math_subscript_token{#3}
    \IfBlankF{#4}
    {
        \exp_last_unbraced:Ne \paren {\IfBooleanT{#1}{*}\IfValueT{#2}{[\exp_not:N #2]}} {#4}
    }
}
\ExplSyntaxOff

\newcommand{\calI}{\mathcal{I}}
\newcommand{\calJ}{\mathcal{J}}
\newcommand{\calL}{\mathcal{L}}

\DeclareMathOperator{\Tr}{Tr}
\newcommand{\SYMdiam}{\operatorname{diam}}
\newcommand{\SYMd}{d}

\DeclareMathOperator{\supp}{supp}

\newcommand{\evol}[2]{\e^{#1}\,#2\,\e^{-#1}}

\newcommand{\Rhalf}{\mathup{R}}
\newcommand{\Uhalf}{\mathup{U}}
\newcommand{\upperboundary}{\mathup{upper}}
\newcommand{\lowerboundary}{\mathup{lower}}
\newcommand{\leftboundary}{\mathup{left}}
\newcommand{\rightboundary}{\mathup{right}}

\newcommand{\supint}{^{\mathup{int}}}

\newcommand{\Cvol}{\mathcal{C}_{\mathup{vol}}}
\newcommand{\Ccod}[1]{\mathcal{C}_{\mathup{vol},#1}}
\newcommand{\graphs}{\mathcal{G}(D,\Cvol)}
\newcommand{\phib}{\phi_\beta}
\newcommand{\phibhat}{\hat{\phi}_\beta}
\newcommand{\alphat}{\tilde{\alpha}}
\newcommand{\Psit}{\tilde{\Psi}}

\newcommand{\Lambdap}{\Lambda^{\smash{\mathup{pert}}}}

\newcommand{\suchthat}{\mathpunct{\ordinarycolon}}

\newcommand{\quadtext}[1]{\quad\text{#1}\quad}
\newcommand{\qquadtext}[1]{\quad\quadtext{#1}\quad}

\newcommand{\Alignindent}{\hspace*{2em}&\hspace*{-2em}}
\providecommand{\mathup}[1]{\mathrm{#1}}

\let\oldsetminus\setminus
\newbox\mybox
\newcommand\cutsetminus[1]{%
    \setbox\mybox\hbox{\(#1\oldsetminus\)}%
    \ht\mybox=0pt%
    \dp\mybox=0pt%
    \usebox\mybox%
}
\renewcommand\setminus{%
    \mathbin{%
        \mathchoice%
            {\displaystyle\oldsetminus}
            {\textstyle\oldsetminus}
            {\cutsetminus{\scriptstyle}}
            {\cutsetminus{\scriptscriptstyle}}
    }%
}

\makeatletter
\ExplSyntaxOn

\newcounter{theoremenv}
\newcounter{theoremenvglobal}
\newlist{thmlist}{enumerate}{1}
\setlist[thmlist]{
    label=\textup{(\roman{thmlisti})},
    ref={(\roman{thmlisti})},
    nosep,
}
\renewcommand{\p@thmlisti}{\perh@ps{\protect\ref{auto-label:\arabic{theoremenvglobal}}}}
\DeclareRobustCommand{\perh@ps}[1]{#1}
\newcommand{\itemref}[1]{%
    \begingroup 
    \let\perh@ps\@gobble\ref{#1}%
    \endgroup
}

\addtotheorempostheadhook{
    \crefalias{thmlisti}{\thmt@envname}
    \setcounter{theoremenv}{\value{\thmt@envname}}
    \stepcounter{theoremenvglobal}
    \exp_args:NNe \renewcommand \thetheoremenv {\use:c{the\thmt@envname}}
    \exp_args:Ne \label {auto-label:\arabic{theoremenvglobal}}
}
\makeatother
\ExplSyntaxOff

\theoremstyle{plain}

\declaretheorem[
    name=Theorem,
]{theorem}

\declaretheorem[
    name=Lemma,
    sibling=theorem,
]{lemma}

\declaretheorem[
    name=Proposition,
    sibling=theorem,
]{proposition}

\theoremstyle{definition}

\declaretheorem[
    name=Assumption,
    sibling=theorem,
    qed=\(\diamond\),
]{assumption}

\theoremstyle{remark}
\declaretheorem[
    name=Remark,
    sibling=theorem,
    qed=\(\diamond\),
]{remark}

%% file: biblatex-publication.tex
\usepackage[
    bibstyle=phys,
    citestyle=numeric,
    biblabel=brackets,
    mincitenames=1,
    maxcitenames=3,
    minalphanames=3,
    maxalphanames=4,
    minbibnames=3,
    maxbibnames=10,
    arxiv=abs,
    eprint=true,
    doi=false,
    url=false,
    sorting=nty,
    useprefix=true,
    alldates=year,
    urldate=long,
]{biblatex}
\renewbibmacro*{note+pages}{%
    \printfield{note}%
    \setunit{\bibpagespunct}%
    \printfield{pages}%
    \newunit%
}
\makeatletter
\DeclareFieldFormat{eprint:arxiv}{%
    \ifhyperref
    {\href{https://arxiv.org/\abx@arxivpath/#1}{%
            arXiv\addcolon\linebreak[2]#1}}
    {arXiv\addcolon
        \nolinkurl{#1}%
    }
}
\makeatother

\renewbibmacro*{maintitle+title}{%
    \iffieldsequal{maintitle}{title}{%
        \clearfield{maintitle}%
        \clearfield{mainsubtitle}%
        \clearfield{maintitleaddon}%
    }{%
        \iffieldundef{maintitle}{}{%
            \printtext[doi/url-link]{%
                \usebibmacro{maintitle}%
            }%
            \newunit\newblock
            \iffieldundef{volume}{}{%
                \printfield{volume}%
                \printfield{part}%
                \setunit{\addcolon\space}%
            }%
        }%
    }%
    \printtext[doi/url-link]{%
        \usebibmacro{title}%
    }%
    \newunit%
}

\DeclareBibliographyDriver{article}{%
    \usebibmacro{bibindex}%
    \usebibmacro{begentry}%
    \usebibmacro{author/translator+others}%
    \setunit{\labelnamepunct}\newblock
    \usebibmacro{title}%
    \newunit
    \printlist{language}%
    \newunit\newblock
    \usebibmacro{byauthor}%
    \newunit\newblock
    \usebibmacro{bytranslator+others}%
    \newunit\newblock
    \printfield{version}%
    \newunit\newblock
    \printtext[doi/url-link]{%
        \usebibmacro{journal+issuetitle}%
        \newunit
        \usebibmacro{byeditor+others}%
        \newunit
        \usebibmacro{note+pages}%
        \newunit\newblock
        \iftoggle{bbx:isbn}
        {\printfield{issn}}
        {}%
    }%
    \newunit\newblock
    \printfield{pubstate}%
    \setunit{\addspace}%
    \printfield{year}%
    \newunit\newblock
    \usebibmacro{doi+eprint+url}%
    \newunit\newblock
    \printfield{addendum}%
    \setunit{\bibpagerefpunct}\newblock
    \usebibmacro{pageref}%
    \newunit\newblock
    \iftoggle{bbx:related}
        {\usebibmacro{related:init}%
        \usebibmacro{related}}
        {}%
    \usebibmacro{finentry}%
}